\documentclass[a4paper,onecolumn,11pt,accepted=2025-01-06]{quantumarticle}
\pdfoutput=1
\usepackage[english]{babel}
\usepackage{physics}
\usepackage{braket}
\usepackage{dsfont}
\usepackage[table,xcdraw]{xcolor}
\usepackage{multirow}
\usepackage{graphicx}
\usepackage{tikz}
\usepackage{amsthm}
\usepackage{amssymb}
\usepackage{fontawesome}
\usetikzlibrary{matrix}
\usepackage{stmaryrd}
\usepackage{makecell}

\usepackage{mathabx}

\usepackage{tikz-cd}
\usepackage[numbers,sort]{natbib}

\usepackage[margin=1in]{geometry}
\usepackage[pagebackref, colorlinks = true, linkcolor = blue, urlcolor  = blue, citecolor = red]{hyperref}

\renewcommand{\epsilon}{\varepsilon}
\renewcommand{\phi}{\varphi}
\newcommand{\M}{\mathcal{M}}

\newcommand{\complex}{\mathbb{C}}
\newcommand{\nat}{\mathbb{N}}

\newcommand{\id}{\mathds{1}}

\newcommand{\sa}[1]{\M(\complex)^{\mathrm{sa}}_{#1}}
\newcommand{\st}{\mathfrak{S}}
\newcommand{\eff}{\mathfrak{E}}
\newcommand{\meas}{\mathfrak{M}}
\newcommand{\ch}{\mathfrak{C}}
\newcommand{\ins}{\mathfrak{I}}
\newcommand{\qs}[1]{\st(\complex^{#1})}
\newcommand{\qe}[1]{\eff(\complex^{#1})}
\newcommand{\qm}[2]{\meas(#1, \complex^{#2})}
\newcommand{\qmall}[1]{\meas(\complex^{#1})}
\renewcommand{\qc}[2]{\ch(\complex^{#1},\complex^{#2})}
\newcommand{\qi}[3]{\ins(#1, \complex^{#2}, \complex^{#3})}

\makeatletter
\newtheorem*{rep@theorem}{\rep@title}
\newcommand{\newreptheorem}[2]{
\newenvironment{rep#1}[1]{
 \def\rep@title{#2 \ref{##1}}
 \begin{rep@theorem}}
 {\end{rep@theorem}}}
\makeatother

\newmuskip\pFqmuskip

\newcommand*\pFq[6][8]{
  \begingroup
  \pFqmuskip=#1mu\relax
  \mathchardef\normalcomma=\mathcode`,
  \mathcode`\,=\string"8000
  \begingroup\lccode`\~=`\,
  \lowercase{\endgroup\let~}\pFqcomma
  {}_{#2}F_{#3}{\left[\genfrac..{0pt}{}{#4}{#5};#6\right]}
  \endgroup
}
\newcommand{\pFqcomma}{{\normalcomma}\mskip\pFqmuskip}

\newtheorem{thm}{Theorem}[section]
\newtheorem*{thm*}{Theorem}
\newreptheorem{thm}{Theorem}

\newtheorem{cor}[thm]{Corollary}
\newtheorem{ex}[thm]{Example}
\newtheorem{defi}[thm]{Definition}
\newtheorem*{defi*}{Definition}
\newtheorem{prop}[thm]{Proposition}
\newtheorem{remark}[thm]{Remark}

\newtheorem{example}[thm]{Example}

\begin{document}

\author{Andreas Bluhm}
\affiliation{Univ.\ Grenoble Alpes, CNRS, Grenoble INP, LIG, 38000 Grenoble, France}
\email{andreas.bluhm@univ-grenoble-alpes.fr}
\orcid{0000-0003-4796-7633}

\author{Leevi Leppäjärvi}
\affiliation{RCQI, Institute of Physics, Slovak Academy of Sciences, Dúbravská cesta 9, 84511 Bratislava, Slovakia}
\email{leevi.leppajarvi@savba.sk}
\orcid{0000-0002-9528-1583}

\author{Ion Nechita}
\affiliation{Laboratoire de Physique Th\'eorique, Universit\'e de Toulouse, CNRS, UPS, France}
\email{ion.nechita@univ-tlse3.fr}
\orcid{0000-0003-3016-7795}

\title{On the simulation of quantum multimeters}

\begin{abstract}
In the quest for robust and universal quantum devices, the notion of simulation plays a crucial role, both from a theoretical and from an applied perspective. In this work, we go beyond the simulation of quantum channels and quantum measurements, studying what it means to simulate a collection of measurements, which we call a multimeter. To this end, we first explicitly characterize the completely positive transformations between multimeters. However, not all of these transformations correspond to valid simulations, as otherwise we could create any resource from nothing. For example, the set of transformations includes maps that always prepare the same multimeter regardless of the input, which we call trash-and-prepare. From the perspective of an experimenter with a given multimeter as part of a complicated setup, having to discard the multimeter and using a different one instead is undesirable. We give a new definition of multimeter simulations as transformations that are triviality-preserving, i.e., when given a multimeter consisting of trivial measurements they can only produce another trivial multimeter. In the absence of a quantum ancilla, we then characterize the transformations that are triviality-preserving and the transformations that are trash-and-prepare. Finally, we use these characterizations to compare our new definition of multimeter simulation to three existing ones: classical simulations, compression of multimeters, and compatibility-preserving simulations.
\end{abstract}

\maketitle

\titlepage

\tableofcontents

\newpage

\section{Introduction}
One of the most important goals in the development of quantum computers is the simulation of quantum systems of interest, both in an analogue and digital fashion \cite{cirac2012goals}. Measurements play an important role in this effort, as they are the only way for us to access the information of a quantum system. Therefore, in this work we focus on the simulation of measurements (or meters) and in particular simulation of collections of measurements, which we call \emph{multimeters}. An extreme case of such a simulation are \emph{compatible} measurements. Measurements are compatible if there exists a joint measurement one can perform instead and obtain the measurement statistics from classical postprocessing, possibly using randomness, of its outcomes \cite{Heinosaari2016,Guhne2023}. In that sense, the joint measurements simulate the compatible measurements, such that it makes more sense to implement the joint measurement than all of the compatible ones. A striking feature of quantum mechanics which distinguishes it from classical mechanics is the existence of incompatible measurements \cite{Heisenberg1927,Bohr1928}, for example projective measurements with non-commuting elements. Therefore, measurement compatibility does not capture the full picture of what it means to simulate a multimeter.

Several different notions have been proposed in the literature for the simulation of multimeters, all of them generalizing measurement compatibility in different ways. The first way in which multimeters can simulate other multimeters is by classical means, see e.g. \cite{guerini2017operational, oszmaniec2017simulating, filippov2018simulability,Oszmaniec2019}. In these works, a collection of measurements can be classically simulated by other measurements by randomly selecting measurements from the simulating set and then classically postprocessing their outcomes. Thus, one can for example ask the question whether a given collection of measurements can be performed using a smaller number of measurements or measurements with less outcomes, thereby simplifying the task. Compatible measurements are then the measurements which can be simulated from one measurement alone and are in a sense as simple as possible. We call this scenario \emph{classical simulation} of multimeters. 

Instead of using postprocessing, we can instead consider a simulation in which the quantum state to be measured can be preprocessed with the help of a quantum instrument to reduce the dimension of the quantum input, partially converting it to classical information. The original collection of measurements one would like to perform is then simulated by performing measurements on this smaller quantum system, possibly using the classical side information in the process.  We call this scenario the \emph{compression} scenario. It has recently been considered in \cite{IoannouSimulability2022,Jones2023,Jokinen2023}. Again, compatible measurements represent an extreme case of this procedure: instead of conserving any quantum system, the joint measurement is performed on the quantum input, thereby destroying it completely. The simulation now consists of classical postprocessing of the outputs of the joint measurements, thereby obtaining the desired output statistics. The simulation of compatible measurements is therefore also in this framework as simple as possible. In addition to generalizing compatibility, compressibility (also called high-dimensional simulability) is shown to be equivalent to high-dimensional steering \cite{Jones2023}.

Finally, we can combine pre- and postprocessing in order to simulate multimeters. This has been done in \cite{buscemi2020complete}. Here, the authors argue that a simulation should preserve the compatibility of measurements, i.e., a multimeter consisting of compatible measurements can only be used to simulate compatible measurements. Note that although the simulation scheme put forward in \cite{buscemi2020complete} preserves compatibility, the authors do not claim that it is the most general scheme which has this property. We call their setup \emph{compatibility-preserving simulation}.

As we have demonstrated, there is no single agreed-upon notion of what it means to simulate a multimeter by another one. The first two notions of simulations are clearly incomparable, as one is only interested in the number of measurements and the classical information resulting from them, whereas the other focuses on the dimension of the quantum input, treating the classical information practically as free. One way of unifying them could be to simply combine them, using compression of the quantum input and classical postprocessing of the output to simulate. In some sense, this is what compatibility-preserving simulation does, but one can imagine even broader notions of multimeter simulation. 

Therefore, our aim is to find the most general definition possible of what it means to  simulate a multimeter by another multimeter. Our article is organized as follows. In Sec.\  \ref{sec:main-results}, we present the main results of our work. In Sec.\ \ref{sec:fundamental-quantum-devices} we collect the necessary preliminaries concerning states, measurements and channels 
in quantum mechanics and set up our framework to study multimeters. Next, in Sec.\  \ref{sec:transformations}, we characterize the transformations between multimeters as quantum supermaps. In order to understand these transformations better, we discuss in Sec.\  \ref{sec:previous-simulations} which kind of operations are encompassed by these transformation between multimeters, including different notions of multimeter simulation that have been proposed in previous work. Subsequently, we give in Sec.\  \ref{sec:multimeter-simulation} our new definition of which transformations should be allowed for a non-trivial notion of multimeter simulation, arguing that multimeter simulations should be the triviality-preserving transformations. In the absence of a quantum ancilla, we characterize the triviality-preserving transformations and the transformations which act as trash-and-prepare maps, i.e., which always simulate the same multimeter. We conclude the section by comparing our results to the previous notions of multimeter simulation. Finally, we end with an outlook in Sec.\  \ref{sec:discussion}.

\section{Main results} \label{sec:main-results}
In this section, we present the main results of this work. The objects we are concerned with are collections of measurements, called \emph{multimeters}. A multimeter $M$ is therefore a (finite) set $\{M_{\cdot|x}\}_{x \in [g]}$ of positive operator-valued measures (POVMs) $M_{\cdot|x} = \{M_{a|x} \}_{a \in [k]}$. A multimeter can be seen as a quantum channel where all the classical information about the measurements involved and their outcomes can be embedded in suitable quantum systems. 

Our first result is a characterization of transformations between multimeters as quantum channels. The transformations we allow are completely positive maps that map Choi matrices of multimeters to Choi matrices of other multimeters. The informal version of our result is the following (see Thm.\  \ref{thm:multimeter-transformation-realization} for the formal version):

\begin{thm*}
For any transformation $\Psi$ which maps multimeters of $g$ POVMs each with $k$ outcomes on a $d$-dimensional quantum system to multimeters of $r$ POVMs each with $l$ outcomes on an $n$-dimensional quantum system, there exist an ancillary system $\complex^s$, completely positive maps $\Lambda_{x|y}$ which form an instrument for any choice of $y$, and a set of POVMs $B = \{B_{\cdot|a,x,y} \}_{a \in [k], x \in [g], y \in [r]}$ such that $\Psi$ maps 
\begin{equation*}
   (M, y) \mapsto \sum_{x=1}^g \sum_{a=1}^k \Lambda^*_{x|y}(M_{a|x} \otimes B_{\cdot |a,x,y})\, .
\end{equation*}
In the Schrödinger picture, this means that the simulated POVMs $\{N_{\cdot|y}\}_{y \in [r]}$ arise from $M$ as 
\begin{equation*}
   \Tr[N_{b|y} \varrho] = \sum_{x=1}^g \sum_{a=1}^k \Tr\left[B_{b |a,x,y} \Tr_{\complex^d}[ (M_{a|x} \otimes \id_s)  \Lambda_{x|y}(\varrho)] \right]
\end{equation*}
for all quantum states $\varrho$.
\end{thm*}
That is, to simulate the multimeter $N$ on an input state, first a conditional instrument $\{\Lambda_{x|y}\}$ is performed on the state that might depend on the measurement $y$ to be simulated. The choice of $y$ means that the multimeter $N$ performs the measurement $\{N_{\cdot|y}\}$ on the input state. The instrument  has a classical outcome $x$ and outputs a quantum system. Depending on its outcome $x$, the measurement $M_{\cdot|x}$ is performed on part of the quantum output of the instrument, giving classical outcome $a$. Finally, another measurement $B_{\cdot|a,x,y}$ is performed on the remaining quantum system, which can depend on all the available classical information. We have illustrated this in Fig.\ \ref{fig:multimeter-simulation-intro}.

\begin{figure}[htb]
    \centering
    \includegraphics{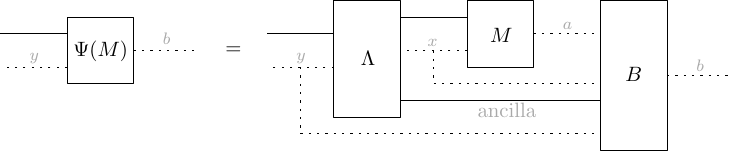}
    \caption{A multimeter $M$ is transformed using instruments $\Lambda_{\cdot|y}$ and a postprocessing $B_{\cdot|a,x,y}$. Quantum systems are depicted by solid lines, while classical systems are represented by dotted wires. Note the \emph{quantum ancilla} wire connecting the multi-instrument $\Lambda$ and the multimeter $B$.}
    \label{fig:multimeter-simulation-intro}
\end{figure}

However, not every transformation between multimeters can be considered a simulation of one multimeter by another. The choice of which transformations one would like to rule out depends on the application one has in mind. In this work, we take the perspective of an experimenter who has a multimeter at her disposal and wonders which other multimeters she can implement with it. Therefore, she wants to use at least some part of the simulating devices in the simulation process instead of just ignoring them. Hence, the transformations we want to rule out are the ones in which any multimeter is replaced by the same fixed multimeter. We call such transformations \emph{trash-and-prepare}, because they just throw away the multimeter and replace it by another. In the case when the ancilla in Fig.\ \ref{fig:multimeter-simulation-intro} is classical, we can characterize these trash-and-prepare transformations. In this case, the instrument has another classical output $\lambda$ and the measurements $B_{\cdot|a,x,y}$ are just a collection of probability distributions $\nu= \{\nu_{\cdot|a,x,y, \lambda}\}_{a ,x ,y, \lambda }$. We have depicted this in Fig.\ \ref{fig:ancilla-free-simulation-intro}.
\begin{figure}[htb]
    \centering
    \includegraphics{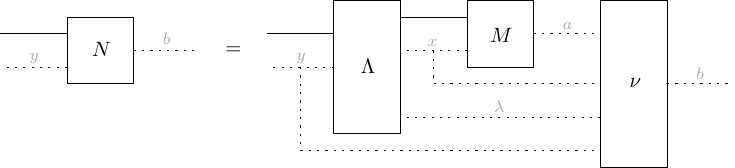}
    \caption{The simulation of multimeter $N$ by the multimeter $M$ admits a realization with a classical ancilla represented by $\lambda \in [s]$. Compare with the general case in Fig.\ \ref{fig:multimeter-simulation-intro}, and notice that in this case the postprocessing $\nu$ and the ancilla $\lambda$ are classical.}
    \label{fig:ancilla-free-simulation-intro}
\end{figure}

We find in Thm.\  \ref{thm:ancilla-free-tap-maps}:
\begin{thm*}
    Let us consider a quantum superchannel $\Psi$ between multimeters that admits a realization $(s, \Lambda, \tilde \nu)$ with a classical ancilla. Then $\Psi$ is trash-and-prepare if and only there exists a possibly different realization $(s, \Lambda, \nu)$ such that all the conditional probability distributions in $ \nu= \{ \nu_{\cdot|a,x,y, \lambda}\}_{a ,x ,y, \lambda}$ are independent of $a$. If $s=1$ (there is not even a classical ancilla), we can take $\nu= \tilde \nu$.
\end{thm*}

Our result can be intuitively illustrated: a transformation $\Psi$ (that admits a realization with a classical ancilla) is trash-and-prepare if and only if there exists a realization as in Fig.\ \ref{fig:ancilla-free-simulation-intro} such that the classical wire $a$ between $M$ and $\nu$ can be cut without changing the map. If this is the case the outcome $a$ of the multimeter $M$ can be simply discarded and the postprocessing $\nu$ is not affected by $a$. Thus, in the end a fixed multimeter is applied irrespective of the input multimeter $M$. We note that our result is constructive so that it also gives the recipe for the (possibly different) postprocessing $\nu$.

Finally, we introduce in our article our definition of what it means to simulate a multimeter by another one. To this end, we consider multimeters of \emph{trivial measurements} which do not depend on the input quantum states, i.e., where $M_{a|x} = p_{a|x} \mathds{1}$ for some probability distributions $p_{\cdot|x}$. We want to call a simulation a transformation that cannot map trivial multimeters to non-trivial ones. Our reasoning is that trivial multimeters discard the quantum state without measuring it. Hence if a transformation maps trivial to non-trivial multimeters, it means that at least one additional device that extracts information from the quantum state is needed. This argument is similar to the idea behind compatibility-preserving simulations, but the property of the multimeters we seek to preserve is much more basic.
\begin{defi*}[Simulation of multimeters]
     A \emph{simulation of multimeters} is a transformation between multimeters, i.e., a quantum superchannel between multimeters, that is \emph{triviality-preserving} in the sense that whenever the input multimeter consists of only trivial POVMs, then the multimeter simulated by $\Psi$ corresponds to a multimeter that only consists of trivial POVMs as well.
\end{defi*}

In the case where the transformation $\Psi$ has an ancilla-free realization (i.e., $s=1$), we can characterize the transformations that are simulations of multimeters in the sense of this new definition. The following result can be found as part of Thm.\  \ref{thm:ancilla-free-tp-maps}:

\begin{thm*}
    Let us consider a quantum superchannel $\Psi$ between multimeters that admits an ancilla-free realization $(\tilde \Lambda, \nu)$. Then $\Psi$ is triviality-preserving if and only if $\Psi$ admits a possibly different ancilla-free realization $(\Lambda, \nu)$ such that $\Lambda_{x | y} = \pi_{x|y} \Phi_{x,y}$ for all $x ,y$ for some conditional probability distribution $\pi = (\pi_{\cdot|y})_{y}$ and a family of \emph{quantum channels} $\{\Phi_{x,y}\}_{x ,y}$.
\end{thm*}

\begin{figure}[htb]
    \centering
    \includegraphics{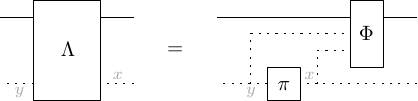}
    \caption{A multi-instrument $\Lambda$ that factorises and induces a triviality-preserving multimeter transformation $\Psi$.}
    \label{fig:Lambda-factorises-triviality-preserving-intro}
\end{figure}

Our result is intuitively illustrated in Fig.~\ref{fig:Lambda-factorises-triviality-preserving-intro}: a tranformation $\Psi$ (that admits an ancilla-free realization) is triviality-preserving if and only if there exists a realization $(\Lambda, \nu)$ such that the preprocessing part $\Lambda$ factorizes as in  Fig.~\ref{fig:Lambda-factorises-triviality-preserving-intro} into just probabilistically applying some set of channels instead of some general instruments.

In conclusion, we have introduced the triviality-preserving transformations between multimeters as the most general reasonable definition of quantum simulation and fully characterized such simulations in the ancilla-free case. This article is therefore the starting point for the further exploration of such simulations of multimeters, especially when we allow a quantum ancilla.

\section{Fundamental quantum devices} \label{sec:fundamental-quantum-devices}

Quantum theory is an operational theory meaning that it can be described by its primitives which are physical devices: \emph{state preparators}, \emph{measurement devices} and \emph{transformations}.  Together they can be used to conduct physical experiments giving information about the systems described by the theory. We will start by recalling the mathematical description of these primitives in quantum theory. See \cite{Heinosaari2011, watrous2018theory} for a more detailed introduction to the formalism.

\subsection{States, measurements, transformations}

Let $d \in \nat$ and let us denote $[d] := \{1, \ldots, d\}$.  We denote the set of complex $d \times d$ matrices by $\M(\complex)_d$ and its subset of self-adjoint (Hermitian) matrices by $\sa{d}$. The states $\qs{d}$ of a $d$-dimensional quantum system are represented by the set of positive-semidefinite matrices in $\sa{d}$ with trace one, i.e., 
\begin{equation}
\qs{d} := \{ \varrho \in \sa{d} \, : \, \varrho \geq 0, \Tr[\varrho] =1 \} \, .
\end{equation}
The elements in $\qs{d}$ are also called \emph{density matrices}. 

A transformation between two quantum systems with density matrices $\qs{d}$ and $\qs{n}$, respectively, is described by a \emph{(quantum) channel} $\Phi$ which is taken to be a \emph{completely positive (CP)} and \emph{trace-preserving (TP)} linear map $\Phi: \M(\complex)_d \to \M(\complex)_n$ meaning that $\Phi \otimes id_{d'}: \M(\complex)_d \otimes \M(\complex)_{d'} \to \M(\complex)_n \otimes \M(\complex)_{d'}$ is positive for all $d' \in \nat$ (CP) and that $\Tr[\Phi(X)] = \Tr[X]$ for all $X \in \M(\complex)_d$ (TP). The set of quantum channels between systems $\qs{d}$ and $\qs{n}$ is denoted by $\qc{d}{n}$. A completely positive map $\Psi: \M(\complex)_d \to \M(\complex)_n$ that is not trace-preserving but only trace-nonincreasing (TNI), i.e. $\Tr[\Psi(X)] \leq \Tr[X]$ for all $X \in  \M(\complex)_d$, is called a \emph{(quantum) operation} and is interpreted as a probabilistic transformation where the transformation probability of a state $\varrho \in \qs{d}$ is given by $\Tr[\Psi(\varrho)]$. 

Measurements on a $d$-dimensional quantum system can be described by using \emph{effect operators}, i.e., positive elements in $\sa{d}$ bounded above by the identity matrix $\id_d$ so that the set of effects $\qe{d}$ is then
\begin{equation}
\qe{d} := \{ E \in \sa{d} \, : \, 0 \leq E \leq \id_d \} \, .
\end{equation}
A measurement (or a meter) with $k \in \nat$ outcomes (where we assume that $k < \infty$ for simplicity) now corresponds to a \emph{positive operator-valued measure (POVM)} $M: j \mapsto M_j$ from $[k]$ to the set of effects $\qe{d}$ such that $\sum_{j =1}^k M_j = \id_d$. The set of $k$-outcome POVMs on $\qs{d}$ is denoted by $\qm{k}{d}$ and the set of all POVMs (with finite outcomes) on $\qs{d}$ is denoted by $\qmall{d}$.  The probability that an outcome $j \in [k]$ is obtained in a measurement of a POVM $M \in \qm{k}{d}$ on a quantum system in state $\varrho \in \qs{d}$ is given by the Born rule as $\Tr[M_j \varrho]$. 

If we consider several measurements, usually not all of them can be measured at the same time, for example, if the effects consist of projections which do not commute. If simultaneous measurement is possible, the measurements are called \emph{compatible} or \emph{jointly measurable} (see \cite{Heinosaari2016} and \cite{Guhne2023} for reviews on joint measurability).
\begin{defi} \label{def:compatible-measurements}
Let $\{E_{\cdot|x}\}_{x \in [g]} \subset \qm{k}{d}$ form a collection of POVMs. These POVMs are \emph{compatible} or \emph{jointly measurable} if there is some $\Lambda \in \mathbb N$ and a POVM $M \in \qm{\Lambda}{d}$ such that
\begin{equation*}
E_{a|x} = \sum_{\lambda=1}^\Lambda p_{a|x,\lambda}M_\lambda
\end{equation*}
for all $x \in [g]$, $a \in [k]$ and some conditional probability distribution $p:= (p_{\cdot|x,\lambda})_{x \in [g], \lambda \in [\Lambda]}$ on $[k]$.
\end{defi}

The interpretation behind compatibility as performing a joint measurement comes from the concept of postprocessing:
\begin{defi}
    A POVM $N \in \qm{l}{d}$ is said to be a \emph{postprocessing} of a POVM $M \in \qm{k}{d}$ if there exists a conditional probability distribution $\mu := (\mu_{\cdot|a})_{a \in [k]}$ on $[l]$ such that $N_b = \sum_{a=1}^k\mu_{b|a} M_a$ for all $b \in [l]$. In this case we denote that $N = \mu \circ M$.
\end{defi}
The interpretation of postprocessing is that if we measure $M$ and obtain an outcome $a$ then  $\mu_{b|a}$ describes the probability of assigning an outcome $b$ instead. Thus, postprocessing describes a classical manipulation of measurement outcomes including merging and splitting different outcomes. Hence, for compatible POVMs  $\{E_{\cdot|x}\}_{x \in [g]} \subset \qm{k}{d}$ we can always find a joint POVM $M \in \qm{\Lambda}{d}$ for some $\Lambda \in \nat$ from which every POVM $E_{\cdot|x}$ can be postprocessed with the conditional probability distributions $p^{(x)} := (p_{\cdot|x,\lambda})_{\lambda \in \Lambda}$ so that $E_{a|x} = (p^{(x)} \circ M)_{a}$ for all $x \in [g]$ and $a \in [k]$.

A measurement device which does not only produce a classical measurement outcome (as measurements described by POVMs do) but also includes the description of the transformation of the measured state is described by a \emph{(quantum) instrument}. A $k$-outcome quantum instrument between $\qs{d}$ and $\qs{n}$ is an operation-valued measure $\Lambda: j \mapsto \Lambda_j$ from $[k]$ to the set of quantum operations between $\qs{d}$ and $\qs{n}$ such that $\Phi^{\Lambda} := \sum_{j=1}^k \Lambda_j$ is a quantum channel in $\qc{d}{n}$. If the system is initially in a state $\varrho \in \qs{d}$, then the (unnormalized) conditional postmeasurement state is $\Lambda_j(\varrho)$, where $j$ is the outcome obtained in the measurement of the \emph{induced POVM} $M^\Lambda$ defined as $\Tr[M^\Lambda_j \varrho] = \Tr[\Lambda_j(\varrho)]$ for all $\varrho \in \qs{d}$. Thus, $M^\Lambda_j = \Lambda^\ast_j(\id_n)$, where $\Lambda^\ast_j$ is the dual map of $\Lambda_j$. The set of instruments from $\qs{d}$ to $\qs{n}$ with $k$ outcomes is denoted by $\qi{k}{d}{n}$.

\subsection{Quantum devices as channels}

All of the previously discussed quantum devices have some number of classical and quantum inputs and outputs: a state preparator is a device with no inputs and one quantum output, a channel is a device with one quantum input and output, the measurement of a POVM corresponds to a device with quantum input and a classical output and an instrument takes a quantum input and produces both a quantum and a classical output. These are depicted in Fig.~\ref{fig:devices}.

\begin{figure}
    \centering
    \includegraphics{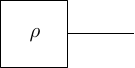}\qquad
    \includegraphics{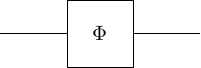}\qquad
    \includegraphics{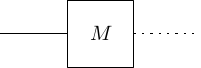}\qquad
    \includegraphics{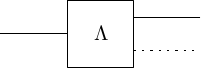}
    
    \caption{Pictorial representation of quantum devices. From left to right: a quantum state (preparator), a quantum channel, a measurement, and an instrument. Diagrams are to be read from left to right. Quantum systems are depicted by solid lines, while classical systems are represented by dotted lines.}
    \label{fig:devices}
\end{figure}

For our purposes it is convenient to consider all of them as channels where the additional knowledge that some of the inputs or outputs are classical gives us more constraints on the specific structure of the channel. In order to keep the mathematical treatment similar for all of the devices, we consider the classical systems to be embedded in a quantum system in the usual way: if $[k]$ is the classical set of indices, then the classical input/output $j \in [k]$ can be described by $\ketbra{j}{j} \in \qs{k}$, where now $\{\ket{j}\}_{j=1}^k$ is some orthonormal basis in $\complex^k$. By possibility of probabilistic mixing, the set of classical states must be convex and thus the most general description of a classical state $\delta \in \qs{k}$ is $\delta = \varrho_q := \sum_{j=1}^k q_j \ketbra{j}{j}$ for some probability distribution $q := (q_j)_{j=1}^k$ on $[k]$. 

When we apply this to the previously introduced devices which have classical outputs, namely POVMs and instruments, we have the following correspondence: we identify a POVM $M\in \qm{k}{d}$ with its related \emph{quantum-classical (q-c) channel} $\Phi_M \in \qc{d}{k}$ defined as 
\begin{equation}
\Phi_M(\varrho) := \sum_{j=1}^k \Tr[M_j \varrho] \ketbra{j}{j}
\end{equation}
for all $\varrho \in \qs{d}$, and similarly an instrument $\Lambda \in \qi{k}{d}{n}$ with the related block-diagonal channel $\Psi_\Lambda \in \qc{d}{kn}$ defined as 
\begin{equation}
\Psi_\Lambda(\varrho) := \sum_{j=1}^k \ketbra{j}{j} \otimes \Lambda_j(\varrho) 
\end{equation}
for all $\varrho \in \qs{d}$, where in both cases the classical information can be read by measuring the classical part of the system in the fixed basis $\{\ket{j}\}_{j=1}^k$.

\subsubsection{Multimeters as quantum channels}

Previously we have described the most fundamental physical devices in quantum theory. However, in the setting of physical experiments we sometimes also want to describe scenarios where we are using different collections of these devices. In this case we can include into the standard description of some collection of devices an additional classical input which can be used to determine which device from the collection is to be used.  In particular, we are interested in measurement devices described by a collection of POVMs such that by providing the device with a classical input, the device determines which POVM from the collection is measured in each round of the experiment. We call these devices \emph{multimeters}. Formally, a multimeter is just a collection $M=\{M_{\cdot |i}\}_{i=1}^g \subset \qm{k}{d}$ of $g$ POVMs each with $k$ outcomes on a $d$-dimensional quantum system.

Similarly as before we want to embed this additional classical system as well as the output, which is a classical measurement outcome, in the corresponding quantum systems. Thus, as we represented a measurement described by a POVM as a quantum-classical channel, now we wish to consider this kind of multimeter as a channel which quantum and classical input and classical output (a (qc)-c channel).

Motivated by this, we make the following definition:
\begin{defi}\label{def:multimeter-as-quantum-channel}
For a set of POVMs $M=\{M_{\cdot |i}\}_{i=1}^g \subset \qm{k}{d}$ for some $g, k \in \nat$ we define the related \emph{multimeter channel} $\Phi_M \in \qc{dg}{k}$ by setting
\begin{equation}
\Phi_M(\sigma) = \sum_{i=1}^g \sum_{j=1}^k \Tr[( M_{j|i} \otimes \ketbra{i}{i} )(\sigma)]\ketbra{j}{j}
\end{equation}
for all $\sigma \in \qs{dg}$.
\end{defi}
Then $\Phi_M(\varrho \otimes \ketbra{i}{i} ) = \sum_{j=1}^k \Tr[M_{j|i} \varrho] \ketbra{j}{j} = \Phi_{M_{\cdot |i}}(\varrho)$ for all $i \in [g]$ and $\varrho \in \qs{d}$, which means that the measurements corresponding to $\Phi_M$ are uniquely defined. Moreover, a mixture of the POVMs $\{M_{\cdot |i}\}_{i=1}^g$ can be measured by providing the multimeter channel with an input of the form $  \varrho \otimes \varrho_p= \varrho \otimes \left( \sum_{i=1}^g p_i \ketbra{i}{i} \right)  $ for some probability distribution $p = (p_i)_{i=1}^g$ on $[g]$ so that $\Phi_M( \varrho \otimes \varrho_p) = \sum_{j=1}^k \Tr\left[\sum_{i=1}^g p_iM_{j |i} \varrho\right] \ketbra{j}{j} = \Phi_{\sum_i p_i M_{\cdot |i}}(\varrho)$ for all $\varrho \in \qs{d}$. We denote the set of multimeter channels with $g \in \nat$ POVMs each with $k \in \nat$ outcomes on a quantum system of dimension $d \in \nat$ by $\M\M(g,k,d) \subseteq \qc{dg}{k}$. If we would like to consider measurements with different numbers of outcomes, we can just choose $k$ to be the highest number of outcomes and pad the measurements with less outcomes with zero effects. We note that from now on we will use the term multimeter to reference both a collection of POVMs and their related multimeter channels.

More useful, additional properties of quantum channels and instruments (such as the Choi–Jamiołkowski isomorphism and the Stinesping dilation) can be found in Appendix \ref{appendix:A}.

\section{Transformations between multimeters} \label{sec:transformations}
\subsection{Quantum superchannels}
Let us briefly summarize the notion of a (quantum) superchannel (for more details see e.g. \cite{chiribella2008transforming,Chiribella2009,jencova2012generalized}). By a \emph{quantum superchannel} we mean a CP map that maps quantum channels to quantum channels. More specifically, a quantum superchannel between channels in $\qc{d}{n}$ and channels in $\qc{d'}{n'}$ is a CP map $\Psi: \M(\complex)_{nd} \to \M(\complex)_{n'd'}$ such that $\Psi\left(\mathcal{J}(\complex^{nd})\right) \subseteq \mathcal{J}(\complex^{n'd'})$. Then by the Choi–Jamiołkowski isomorphism quantum superchannels $\Psi: \M(\complex)_{nd} \to \M(\complex)_{n'd'}$ are in one-to-one correspondence between CP maps $\hat{\Psi}$ that map linear maps from $\M(\complex)_d$ to $\M(\complex)_n$ to linear maps from $\M(\complex)_{d'}$ to $\M(\complex)_{n'}$ such that $\hat{\Psi}(\qc{d}{n}) \subseteq \qc{d'}{n'}$. It is known that any such map $\hat{\Psi}$ corresponding to a quantum superchannel $\Psi$ can be realized as $\hat{\Psi}(\Phi) = \Psi_{post} \circ \left( \Phi \otimes id_s   \right) \circ \Psi_{pre}$ for all channels $\Phi \in \qc{d}{n}$ for some  preprocessing channel $\Psi_{pre} \in \qc{d'}{ds}$ and a postprocessing channel $\Psi_{post} \in \qc{ns}{n'}$ for some ancillary system $\complex^s$ for some $s \in \nat$. Furthermore, the Choi matrices of quantum superchannels are also called \emph{2-combs} and for a superchannel $\Psi$ with prep- and postprocessing channels $\Psi_{pre}$ and $\Psi_{post}$ it can be shown that $J_{\Psi} = J_{\Psi_{post}} * J_{\Psi_{pre}}$ and $J_{\hat{\Psi}(\Phi)} = J_{\Psi} * J_{\Phi} = J_{\Psi_{post}} * J_{\Phi} * J_{\Psi_{pre}}$ for all $\Phi \in \qc{d}{n}$, where $*$ denotes the link product of Choi matrices. 

Next we will take a closer look on the structure of superchannels that map multimeters to multimeters. 

\subsection{Channels between multimeters}
In this section we want to characterize all possible transformations between multimeters. Since we can represent multimeters as a particular type of quantum channels as in Def.\ \ref{def:multimeter-as-quantum-channel}, we are in particular interested in transformations between quantum channels that describe multimeters. As was explained at the beginning of this section, these type of transformations are represented by quantum superchannels.  We will now proceed to give an elementary realization results for these superchannels on multimeters. A comparison to previous realization results is considered in Remark \ref{remark:comparison}.

We recall that the set of multimeters (as defined in Def.\ \ref{def:multimeter-as-quantum-channel}) with $g \in \nat$ POVMs each with $k \in \nat$ outcomes on a quantum system of dimension $d \in \nat$ is denoted by $\M\M(g,k,d) \subseteq \qc{dg}{k}$. Thus, we want to look at transformations between $\M\M(g,k,d)$ and $\M\M(r,l,n)$ for some fixed $d,n,g,r,k,l \in \nat$. Such transformations are represented by CP maps $\hat{\Psi}$ which map linear maps from $\M(\complex)_{dg}$ to $\M(\complex)_k$ to linear maps from $\M(\complex)_{nr}$ to $\M(\complex)_{l}$ such that $\hat{\Psi}(\M\M(g,k,d)) \subseteq \M\M(r,l,n)$. Let us denote the Choi matrices of multimeters in $\M\M(g,k,d)$ by $\mathcal{J}(\M\M(g,k,d))$. Because of the Choi-Jamiołkowski isomorphism, such maps $\hat{\Psi}$ correspond to CP maps $\Psi: \M(\complex)_{kdg} \to \M(\complex)_{lnr}$ such that $\Psi(\mathcal{J}(\M\M(g,k,d))) \subseteq \mathcal{J}(\M\M(r,l,n))$, where the correspondence is given by
\begin{equation}\label{eq:mm-correspondence}
    \hat{\Psi}(\Phi) = \mathcal{E}_{\Psi(J_\Phi)}, \quad \Psi(J_\Phi) = J_{\hat{\Psi}(\Phi)}
\end{equation}
for all $\Phi: \M(\complex)_{dg} \to \M(\complex)_k$, where $J_\Phi$ is the Choi matrix of map $\Phi$ and $\mathcal{E}_{J}$ is the inverse map defined by a Choi matrix $J$ (see Appendix \ref{appendix:A}, Eqs.\ \eqref{eq:Choi-matrix} and \eqref{eq:Choi-map}).

Now we can show the following realization theorem (the proof can be found in Appendix \ref{appendix:B}):

\begin{thm} \label{thm:multimeter-transformation-realization}
Let $\Psi: \M(\complex)_{kdg} \to \M(\complex)_{lnr}$ be a CP map such that $\Psi(\mathcal{J}(\M\M(g,k,d))) \subseteq \mathcal{J}(\M\M(r,l,n))$. Then $\Psi$ has a realization $(\complex^s, \Lambda, B)$, i.e., there exist an ancillary system $\complex^s$ for some $s \in \nat$, CP maps $\Lambda^*_{x|y}: \M(\complex)_{ds} \to \M(\complex)_n$ such that $\Lambda^*_y := \sum_{x \in [g]} \Lambda^*_{x|y}$ is a unital CP (UCP) map for all $y \in [r]$, and a set of POVMs $B = \{B_{\cdot|a,x,y} \}_{a \in [k], x \in [g], y \in [r]} \subset \qm{l}{s}$ such that 
\begin{equation}
    \Psi \left(J_{\Phi_{\left\lbrace M_{\cdot|x}\right\rbrace_{x\in [g]}}} \right) = J_{\Phi_{\left\lbrace \sum_{x=1}^g \sum_{a=1}^k \Lambda^*_{x|y}(M_{a|x} \otimes B_{\cdot |a,x,y}) \right\rbrace_{y \in [r]}}},   
\end{equation}
where $\left\lbrace \sum_{x=1}^g \sum_{a=1}^k \Lambda^*_{x|y}(M_{a|x} \otimes B_{\cdot |a,x,y}) \right\rbrace_{y \in [r]} \subset \qm{l}{n}$ is a set of POVMs for all \\ $\{M_{\cdot|x}\}_{x\in [g]} \subset \qm{k}{d}$.
\end{thm}

If now $\Psi: \M(\complex)_{kdg} \to \M(\complex)_{lnr}$ is a CP map characterized by the previous theorem, then translating back to the CP map $\hat{\Psi}$ given by Eq.\ \eqref{eq:mm-correspondence} which maps linear maps from $\M(\complex)_{dg}$ to $\M(\complex)_k$ to linear maps from $\M(\complex)_{nr}$ to $\M(\complex)_{l}$ we have that for $\sigma \in \mathcal M(\mathbb C)_{nr}$,
\begin{align*}
    \hat{\Psi}\left(\Phi_{\left\lbrace M_{\cdot|x}\right\rbrace_{x\in [g]}}\right)(\sigma) &=  \mathcal{E}_{\Psi \left(J_{\Phi_{\left\lbrace M_{\cdot|x}\right\rbrace_{x\in [g]}}} \right)}(\sigma) = \mathcal{E}_{J_{\Phi_{\left\lbrace \sum_{x=1}^g \sum_{a=1}^k \Lambda^*_{x|y}(M_{a|x} \otimes B_{\cdot |a,x,y}) \right\rbrace_{y \in [r]}}}}(\sigma) \\
    &= \Tr_{\complex^{nr}}\left[\left(\id_l \otimes \sigma^T\right)J_{\Phi_{\left\lbrace \sum_{x=1}^g \sum_{a=1}^k \Lambda^*_{x|y}(M_{a|x} \otimes B_{\cdot |a,x,y})  \right\rbrace_{y \in [r]}}} \right] \\
    &=  \Tr_{\complex^{nr}}\Bigg[\left(\id_l \otimes \sigma^T\right) \Bigg( \sum_{b=1}^l\sum_{y=1}^r \ketbra{b}{b} \otimes \left( \sum_{x=1}^g \sum_{a=1}^k \Lambda^*_{x|y}(M_{a|x} \otimes B_{b |a,x,y}) \right)^T \\ & \qquad\qquad \otimes \ketbra{y}{y} \Bigg)  \Bigg] \\
    &= \sum_{b=1}^l\sum_{y=1}^r  \Tr\left[\sigma^T \left( \left( \sum_{x=1}^g \sum_{a=1}^k  \Lambda^*_{x|y}(M_{a|x} \otimes B_{b |a,x,y}) \right)^T \otimes \ketbra{y}{y} \right)   \right] \ketbra{b}{b} \\
    &= \sum_{b=1}^l\sum_{y=1}^r  \Tr\left[ \left( \left( \sum_{x=1}^g \sum_{a=1}^k  \Lambda^*_{x|y}(M_{a|x} \otimes B_{b |a,x,y}) \right)  \otimes \ketbra{y}{y} \right) \sigma  \right] \ketbra{b}{b} \\
    &= \underbrace{\Phi_{\left\lbrace \sum_{x=1}^g \sum_{a=1}^k  \Lambda^*_{x|y}(M_{a|x} \otimes B_{\cdot |a,x,y})  \right\rbrace_{y\in [r]}}}_{\in \M\M(r,n,l)}(\sigma)
\end{align*}
for all $\{M_{\cdot|x}\}_{x\in [g]} \subset \qm{k}{d}$. 

From the Schrödinger picture we get a clear recipe of how to measure the transformed POVMs $\left\lbrace N_{\cdot|y} := \sum_{x=1}^g \sum_{a=1}^k  \Lambda^*_{x|y}(M_{a|x} \otimes B_{\cdot |a,x,y})  \right\rbrace_{y\in [r]}$ given by $\Psi$:
\begin{align*}
   \Tr[N_{b|y} \varrho]  &= \Tr[\sum_{x=1}^g \sum_{a=1}^k  \Lambda^*_{x|y}(M_{a|x} \otimes B_{b |a,x,y}) \varrho] = \sum_{x=1}^g \sum_{a=1}^k \Tr[ (M_{a|x} \otimes B_{b |a,x,y})  \Lambda_{x|y}(\varrho)] \\
   &= \sum_{x=1}^g \sum_{a=1}^k \Tr\left[B_{b |a,x,y} \Tr_{\complex^d}[ (M_{a|x} \otimes \id_s)  \Lambda_{x|y}(\varrho)] \right]
\end{align*}
for all $\varrho \in \qs{n}$ for all $b \in [l]$ and $y \in [r]$. 

The interpretation is then as follows: given a state $\varrho \in \qs{n}$ if we want to measure the transformed POVM $N_{\cdot|y}$ given by the label $y$, we first apply the quantum instrument $\Lambda_y \in \qi{g}{n}{ds}$ (which is defined by the quantum operations $\Lambda_{x|y}$) from which we obtain an outcome $x$ and the conditional postmeasurement state $\Lambda_{x|y}(\varrho)$ of the system $\complex^d\otimes\complex^s$. Given $x$, we now measure the POVM $M_{\cdot|x}$ on the system $\complex^d$ and obtain an outcome $a$ while simultaneously leaving the system $\complex^s$ untouched (by just applying the identity channel on that part of the system). Now finally given the classical inputs and outputs $y,x,a$ we measure the system $\complex^s$ with a POVM $B_{\cdot|a,x,y}$ and obtain the final outcome $b$ which we report as the outcome of the transformed measurement $N_{\cdot|y}$; see Fig.\ \ref{fig:multimeter-simulation} for a graphical depiction of this procedure. 

\begin{figure}[htb]
    \centering
    \includegraphics{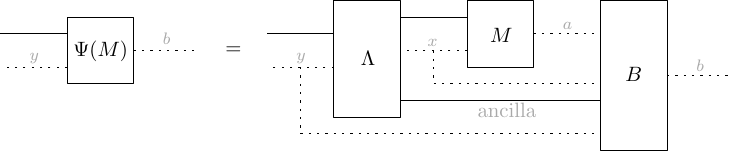}
    \caption{A multimeter $M$ is transformed using instruments $\Lambda_{\cdot|y}$ and a postprocessing $B_{\cdot|a,x,y}$. Quantum systems are depicted by solid lines, while classical systems are represented by dotted wires. Note the \emph{quantum ancilla} wire connecting the multi-instrument $\Lambda$ and the multimeter $B$.}
    \label{fig:multimeter-simulation}
\end{figure}

\begin{remark}\label{remark:comparison}
    Let us comment on how Thm.~\ref{thm:multimeter-transformation-realization} compares to previous results about supermaps from \cite{chiribella2008transforming,gour2019comparison}. One can start from the general setting of quantum supermaps (transforming quantum channels into quantum channels) and impose that some of the quantum systems appearing as inputs and/or outputs are classical. Following \cite{chiribella2008transforming}, this would yield a realization as in Fig.~\ref{fig:general-simulation}. Note that the formulation in Thm.~\ref{thm:multimeter-transformation-realization} and Fig.~\ref{fig:multimeter-simulation} can be mapped into this off-the-shelf form, by encapsulating all the classical systems denoted by the letters $y,x$ into the quantum ancilla between the pre- and post-processing operations $\Lambda$ and $B$. Our precise formulation in Thm.~\ref{thm:multimeter-transformation-realization} has the benefit of keeping separate the classical and the quantum ancillae; this will be useful later when considering simulations that do not use a quantum ancilla. Very recently, after the appearance of this work, in \cite{Allen2024} the authors generalized this formulation to any finite number of classical and quantum inputs and outputs.

\begin{figure}[htb]
    \centering
    \includegraphics{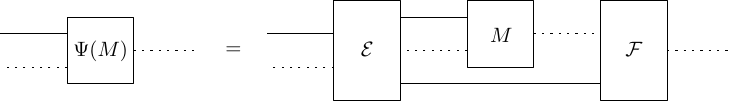}
    \caption{A general simulation of a multimeter $M$ using the supermap formalism from \cite{chiribella2008transforming,gour2019comparison}. The simulation is performed by a quantum pre-processing operation $\mathcal E$, followed by a quantum post-processing operation $\mathcal F$ that is also connected to $\mathcal E$ by a quantum ancilla system. Compare to Fig.~\ref{fig:multimeter-simulation}, where the classical ancillae are explicit.}
    \label{fig:general-simulation}
\end{figure}
\end{remark}

The next example shows that the realization given by Thm.\ \ref{thm:multimeter-transformation-realization} is in general not unique.
\begin{example}\label{ex:t-a-p-not-unique}
Let us consider a transformation $\Psi: \M(\complex)_{kdg} \to \M(\complex)_{ln}$ (here $r=1$) which maps any input multimeter $M = \{M_{\cdot|x}\}_{x \in [g]} \subset \qm{k}{d}$ as $\Psi(J_{\Phi_M}) = J_{\Phi_N}$ where $N \in \qm{l}{n}$ is a fixed trivial POVM  defined by some probability distribution $p$ on $[l]$ as $N_b = p_b \id$ for all $b \in [l]$. Let us define a conditional probability distibution $\nu = (\nu_{\cdot|a,x})_{a \in [k], x \in [g]}$ as $\nu_{b|a,x} = p_b$ for all $b \in [l]$, $a \in [k]$ and $x \in [g]$. It is clear that $(\complex, \Lambda, \nu)$ is a realization of $\Psi$ for \emph{any} $\Lambda \in \qi{g}{n}{d}$: Indeed, we have 
$$\sum_{x=1}^g \sum_{a=1}^k \nu_{b |a,x} \Lambda^*_{x}(M_{a|x}) = \sum_{x=1}^g \sum_{a=1}^k p_b \Lambda^*_{x}(M_{a|x}) = p_b \sum_{x=1}^g  \Lambda^*_{x}(\id) =  p_b \id = N_b $$
for all $b \in [l]$ and all multimeters $M = \{M_{\cdot|x}\}_{x \in [g]} \subset \qm{k}{d}$. We note that in the case of $s=1$ the realization given by Thm.\ \ref{thm:multimeter-transformation-realization} indeed is of the form of the LHS in the previous equation since in that case it is evident that any POVM $A \in \qm{l}{s}$ is actually just a probability distribution $A=(A_b)_{b=1}^l$ on $[l]$, so that in particular the set of POVMs $B = \{B_{\cdot|a,x} \}_{a \in [k], x \in [g]} \subset \qm{l}{s}$ are just conditional probability distributions on $[l]$ which we labeled by $\nu$. Since we can choose any instrument $\Lambda$, this shows that the realization of the map $\Psi$ is not unique.
\end{example}

\section{Previously considered simulation schemes}\label{sec:previous-simulations}

One of the main motiviations behind this work is that there are several versions of simulation of measurements in the literature, but they are in general incomparable. Our aim is therefore to find the most general notion of simulability of multimeters that encompasses all the existing definitions.  Our starting point is that a simulation of any kind of devices (in our case multimeters) is a process that takes some existing device and transforms it to some other device. Thus, we will consider a simulation to be a (specific type of) transformation between multimeters which were characterized in Thm.\ \ref{thm:multimeter-transformation-realization}. Before we move on to considering more general simulations let us first review the previously considered notions of simulations of measurements.

\subsection{Realizations with a classical ancilla}

In what follows we will see that although the existing definitions of simulability are different in their nature they do share one common property as transformations between multimeters: none of them actually utilize the quantum ancilla in the realization scheme in Thm.\ \ref{thm:multimeter-transformation-realization} (see Fig.\ \ref{fig:multimeter-simulation}). If this is the case we call the realization of the transformation a \emph{realization with a classical ancilla} so that classical information is still allowed to be utilized. We note that this only addresses a particular realization of the supermap.

Thus, as we want to allow for a classical ancilla but not a quantum ancilla, what we want to consider a realization with a classical ancilla is a map as in Fig.\ \ref{fig:multimeter-simulation} but where the solid wire for quantum ancilla is replaced by a classical dotted wire. Formally we can do it as follows: In Thm.\ \ref{thm:multimeter-transformation-realization} for the realization $(\complex^s, \Lambda, B)$  of a map $\Psi: \M(\complex)_{kdg} \to \M(\complex)_{lnr}$ that transforms multimeters into multimeters we now take $s$ to represent the size of the classical ancilla system such that the classical ancilla is embedded in the quantum system of dimension $s$. Thus, we assume that only classical information is carried and measured on this ancilla. In particular, this means that the multimeter $B= \{B_{\cdot|a,x,y}\}_{a \in [k], x \in [g], y \in [r]} \subset \qm{l}{s}$ can only be a classical measurement so that it must be just a postprocessing of a measurement that distinguishes the $s$ different classical (pure) states. Hence, if we fix an orthonormal basis $\{ \ket{\lambda}\}_{\lambda=1}^s$ of $\complex^s$ correspoding to the $s$ different classical pure states, then we have that $B_{\cdot|a,x,y} = \sum_{\lambda=1}^s \nu_{\cdot|a,x,y,\lambda} \ketbra{\lambda}{\lambda}$ for all $a \in [k]$, $x \in [g]$ and $y \in [r]$ for some conditional probability distribution $\nu= (\nu_{\cdot|a,x,y,\lambda})_{a \in [k], x \in [g], y \in [r], \lambda \in [s]}$ on $[l]$. Here, $\nu$ then represents the postprocessing of the basis measurement related to the basis $\{ \ket{\lambda}\}_{\lambda=1}^s$ which tranforms it into the multimeter $B$. If we plug in this form of $B$ in the realization $(\complex^s, \Lambda, B)$ given by Thm.\ \ref{thm:multimeter-transformation-realization} we get that the transformed multimeter takes the form
\begin{align*}
    \sum_{x=1}^g \sum_{a=1}^k \Lambda^*_{x|y}(M_{a|x} \otimes B_{\cdot|a,x,y}) = \sum_{x=1}^g \sum_{a=1}^k \sum_{\lambda=1}^s \nu_{\cdot|a,x,y,\lambda} \Lambda^*_{x|y}(M_{a|x} \otimes \ketbra{\lambda}{\lambda})
\end{align*}
for all $y \in [r]$ and  $M=\{M_{\cdot|x}\}_{x \in [g]} \subset \qm{k}{d}$. We can now consider another set of instruments which we also label by $\Lambda$ by defining $\Lambda =\{\Lambda_{\cdot , \cdot |y}\}_{y \in [r]} \subset \qi{g \cdot s}{n}{d}$, where we have set $\Lambda^*_{x, \lambda |y}(A) := \Lambda^*_{x|y}(A \otimes \ketbra{\lambda}{\lambda})$ for all $x \in [g]$, $y \in [r]$, $\lambda \in [s]$ and $A \in \qs{d}$. Now we get that the transformed multimeter takes the form 
\begin{equation}\label{eq:ancilla-free}
    N_b = \sum_{x=1}^g \sum_{a=1}^k \sum_{\lambda=1}^s \nu_{\cdot|a,x,y,\lambda} \Lambda^*_{x, \lambda|y}(M_{a|x}) \subset \qm{l}{n}
\end{equation}
for all $y \in [r]$ and  $M=\{M_{\cdot|x}\}_{x \in [g]} \subset \qm{k}{d}$. Based on this we make the following definition:

\begin{defi}
    A map $\Psi: \M(\complex)_{kdg} \to \M(\complex)_{lnr}$ that transforms multimeters into multimeters has a \emph{realization with a classical ancilla} if there exists $s \in \nat$, a set of instruments $\Lambda =\{\Lambda_{\cdot,\cdot |y}\}_{y \in [r]} \subset \qi{g \cdot s}{n}{d}$ and a conditional probability distribution $\nu= (\nu_{\cdot|a,x,y,\lambda})_{a \in [k], x \in [g], y \in [r], \lambda \in [s]}$ on $[l]$ such that the transformed multimeters corresponding to the Choi matrix $\Psi(J_{\Phi_M})$ take the form of Eq.\ \eqref{eq:ancilla-free} for all $M=\{M_{\cdot|x}\}_{x \in [g]} \subset \qm{k}{d}$. We denote this realization by $(s, \Lambda, \nu)$ or simply $(\Lambda, \nu)$ in the special case when $s=1$.
\end{defi}

The transformation process goes as follows: Given a state $\varrho \in \qs{n}$ and a label $y$ for the resulting measurement, we measure the state with an instrument $\Lambda_{\cdot, \cdot |y} \in \qi{g \cdot s}{n}{d}$, obtain outcomes $x \in [g]$ and $\lambda \in [s]$.  After the measurement we also get a conditional output state $\Lambda_{x, \lambda|y}(\varrho)$, which we then measure by using the POVM $M_{\cdot|x}$ and obtain an outcome $a \in [k]$. Lastly, given the input $y$ and outcomes $x$ and $\lambda$, we postprocess the obtained outcome $a$ to an outcome $b$ with probability $\nu_{b|a,x,y, \lambda}$ and report this as the final outcome of the measurement corresponding to label $y$. Hence, we may identify the instruments $\Lambda$ as the preprocessing part of the transformation and similarly the conditional probabilities $\nu$ as the postprocessing part. We note that here the role of the classical ancilla is just to relay the classical side-information $\lambda$ given by the preprocessing $\Lambda$ and which affects the postprocessing $\nu$. See Fig.\ \ref{fig:ancilla-free-simulation} for an illustration of this process. 

\begin{figure}[htb]
    \centering
    \includegraphics{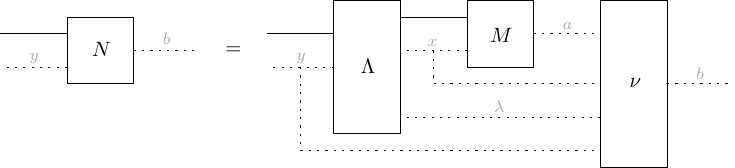}
    \caption{The simulation of multimeter $N$ by the multimeter $M$ admits a realization with a classical ancilla. Compare with the general case in Fig.\ \ref{fig:multimeter-simulation}, and notice that in this case the postprocessing $\nu$ and the ancilla $\lambda$ are classical.}
    \label{fig:ancilla-free-simulation}
\end{figure}

Note that there are multimeter transformations which do require a \emph{quantum} ancilla, as it is demonstrated in the following example.

\begin{example}
    In this example, we shall present a POVM transformation that requires a \emph{quantum ancilla}. To this end, consider a transformation $M \mapsto N$, which takes a POVM $M\in \qm{k}{d}$ to a POVM $N\in \qm{l}{n}$, that is given by
    $$\forall b\in [l]: \qquad N_b = \sum_{a=1}^k \frac{\Tr[M_a]}{d} B_{b|a},$$
    where $\{B_{\cdot|a}\}_{a\in [k]}\subset \qm{l}{n}$ is a collection of incompatible POVMs; see Fig. \ref{fig:quantum-ancilla-needed} for a graphical representation of this transformation. 

    \begin{figure}[htb]
        \centering
        \includegraphics{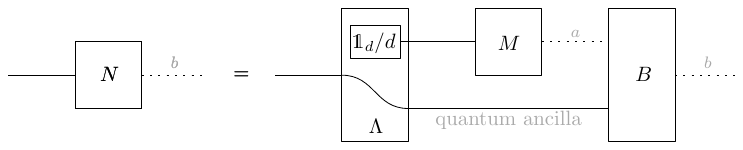}
        \caption{A  POVM transformation that requires a quantum ancilla. The quantum system is passed as an input to an incompatible multimeter $B$ that is controlled by the measurement outcome of the input POVM $M$ when presented with a maximally mixed state.}
        \label{fig:quantum-ancilla-needed}
    \end{figure}   

    To show that this transformation requires a quantum ancilla, let us assume one can realize it only using a classical ancilla (see Fig.\ \ref{fig:quantum-ancilla-needed-only-classical}) as in Eq.\ \eqref{eq:ancilla-free}: 
    $$\forall b\in [l]: \qquad \sum_{a=1}^k \frac{\Tr[M_a]}{d} B_{b|a} = \sum_{a=1}^k \sum_{\lambda=1}^s \Lambda_\lambda^*(M_a) \nu_{b|a, \lambda}.$$
    For all $a \in [k]$ consider now the trivial POVMs $M^{(a)} \in \qm{k}{d}$ with effect $\id_d$ for the outcome $a$ and effect $0$ for all the other outcomes $a'\neq a$. We have: 
    $$\forall a \in [k],b\in [l]: \qquad B_{b|a} = \sum_{\lambda=1}^s \Lambda_\lambda^*(\id_d) \nu_{b|a, \lambda}.$$
    This implies that the POVMs on the LHS of the equation above are postprocessings of the POVM $\Lambda_{\cdot}^*(\id_d) \in \qm{s}{n}$ showing that they are compatible, which is a contradiction.

    \begin{figure}[htb]
        \centering
        \includegraphics{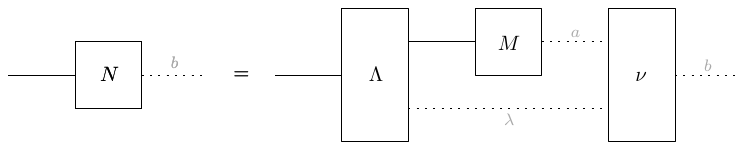}
        \caption{An (impossible) realization of the POVM transformation from Fig.\ \ref{fig:quantum-ancilla-needed} using only a classical ancilla $\lambda$.}
        \label{fig:quantum-ancilla-needed-only-classical}
    \end{figure}    
\end{example}

Later we will also consider an example showing that in some cases although a quantum ancilla may not be needed still a classical ancilla is required for the realization (see Example \ref{ex:cl-ancilla-needed}).

\subsection{Pre- and postprocessings}
In the simplest cases we can consider separately transformations where we either just preprocess or just postprocess. These lead to generic type of transformations that can be considered as simulations in the literature. We will consider the simplest cases without even a classical ancilla ($s=1$).

\subsubsection{Quantum preprocessings}
As a first simple example which could be considered a simulation is the case when we have a realization $(\Lambda, \nu)$ without even a classical ancilla for the quantum supermap $\Psi: \M(\complex)_{kdg} \to \M(\complex)_{lnr}$ such that $l=k$, $r=g$,  $\nu_{b|a,x,y} = \mathbf{1}_{a=b}$ for all $a,b \in [k]$ and $x,y \in [g]$, and $\Lambda^*_{x|y} = \mathbf{1}_{x=y} \Omega^*$ for some fixed channel $\Omega \in \qc{n}{d}$. In this case the transformed POVMs take the form
\begin{equation}
    \sum_{x=1}^g \sum_{a=1}^k \nu_{\cdot |a,x,y} \Lambda^*_{x|y}(M_{a|x}) =\Omega^*(M_{\cdot|y})  \in \qm{k}{n}
\end{equation}
for all $y \in [g]$ all sets of POVMs $M=\{M_{\cdot|x}\}_{x \in [g]} \subset \qm{k}{d}$. Thus, in this simulation scheme, given an input $y$, the measurement $M_{\cdot|y}$ is chosen and it is mapped to a POVM $\Omega^*(M_{\cdot|y})$. In the Schrödinger picture the recipe for doing this is just mapping the state $\varrho \in \qs{n}$ which we wish to measure by the channel $\Omega$ resulting in a state $\Omega(\varrho) \in \qs{d}$ and then just perform the measurement $M_{\cdot|y}$ on this transformed system. 

A preprocessing scheme with slightly more structure can be obtained when we choose $l=k$, $\nu_{b|a,x,y} = \mathbf{1}_{a=b}$, and $\Lambda^*_{x|y} = \Omega^*_{x|y}$ for all $a,b \in [k]$, $x \in [g]$ and $y\in [r]$ for some fixed set of instruments $\Omega=\{\Omega_{\cdot|y}\}_{y \in [r]} \subset \qi{g}{n}{d}$. 
In this case the transformed POVMs take the form
\begin{equation}
    \sum_{x=1}^g \sum_{a=1}^k \nu_{\cdot |a,x,y} \Lambda^*_{x|y}(M_{a|x}) =\sum_{x=1}^g \Omega^*_{x|y}(M_{\cdot|x})  \in \qm{k}{n}
\end{equation}
for all $y \in [r]$ for all sets of POVMs $M=\{M_{\cdot|x}\}_{x \in [g]} \subset \qm{k}{d}$. Thus, in this simulation scheme, a measurement with a label $y$ is obtained first by measuring the input state $\varrho \in \qs{n}$ by the instrument $\Omega_{\cdot|y}$, obtaining an outcome $x$ and bringing the system into a conditional output state $\Omega_{x|y}(\varrho)$, which is then measured by the POVM $M_{\cdot|x}$ from which an outcome $a $ is obtained and then reported as the final outcome of the transformed measurement with a label $y$. Something similar to this type of preprocessing is considered further in Sec.\ \ref{sec:compressibility}. Both types of preprocessing simulations discussed above are represented pictorially in Fig.\ \ref{fig:preprocessing}.

\begin{figure}[htb]
    \centering
    \includegraphics{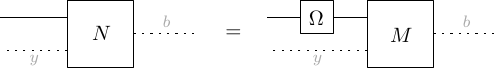} \\ \vspace{.5cm} \includegraphics{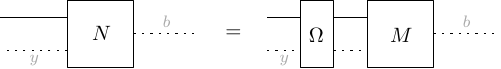}
    \caption{Preprocessing a multimeter $M$ by a quantum device $\Omega$. In the top panel, the device does not depend on the choice of measurement $y$, while in the bottom panel, it can.}
    \label{fig:preprocessing}
\end{figure}

\subsubsection{Classical preprocessings and postprocessings}\label{sec:cl-pre-post}
Another simple special case of a realization $(\Lambda, \nu)$ without even a classical ancilla is when we do not consider the preprocessing to be quantum at all, i.e., we set $n=d$ and $\Lambda^*_{x|y}  = p_{x|y} \mathrm{id}_d$ for some conditional probability distribution $p=(p_{\cdot|y})_{y \in [r]}$ on $[g]$. While we examine the most general case of this later in Sec.\ \ref{sec:classical-simulation}, in the special case when $\nu_{\cdot |a,x,y} = \nu_{\cdot |a,x',y}$ for all $a \in [k]$, $x,x' \in [g]$ and $y \in [r]$, we see that the transformed POVMs are of the form
\begin{equation}
    \sum_{x=1}^g \sum_{a=1}^k \nu_{\cdot |a,x,y} \Lambda^*_{x|y}(M_{a|x}) =\sum_{a=1}^k \nu_{\cdot |a,y} \left( \sum_{x=1}^g  p_{x|y} M_{a|x} \right) \in \qm{l}{n}
\end{equation}
for all $y \in [r]$. Thus, in this simulation scheme, a measurement with a label $y$ is obtained first by measuring the input state $\varrho \in \qs{n}$ with the POVM $M_{\cdot|x}$, where $x$ is first obtained by the conditional probability distribution $p_{\cdot|y}$, after which an outcome $a$ is obtained from the measurement and then it is postprocessed to an outcome $b$ with probability $\nu_{b |a,y}$ and then reported as the final outcome of the transformed measurement with a label $y$; this type of simulation is depicted in Fig.\ \ref{fig:classical-pre-post}. As an important special case we can now obtain mixtures of the POVMs $\{M_{\cdot|x}\}_{x \in [g]}$ by setting $l=k$ and $\nu_{b|a,y} = \mathbf{1}_{a=b}$ for all $a,b \in [k]$ and $y \in [r]$. Then, the resulting POVM is indeed of the form $\sum_{x=1}^g  p_{x|y} M_{\cdot|x}$ for all $y \in [r]$.

\begin{figure}
    \centering
    \includegraphics{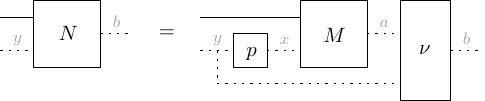}
    \caption{Classical pre- and post- processing of a multimeter $M$. Both devices used in the simulation ($p$ and $\nu$) are classical.}
    \label{fig:classical-pre-post}
\end{figure}

On the other hand, if we do not assume that $\nu$ is independent of $x$, and choose instead that $p_{x|y}=1$ for some $x = x_y \in [g]$, then the resulting POVM takes the following form
\begin{equation}
    \sum_{x=1}^g \sum_{a=1}^k \nu_{\cdot |a,x,y} \Lambda^*_{x|y}(M_{a|x}) =\sum_{a=1}^k \nu_{\cdot |a,x_y,y}  M_{a|x_y}
\end{equation}
for all $y \in [r]$. Thus, the simulated POVMs are simply just classical postprocessings of the original measurements.

\subsection{Classical simulation} \label{sec:classical-simulation}
From the more specific notions of simulations, let us start by reviewing a classical notion of simulation, purely in terms of classical mixing and postprocessing of measurements (see e.g. \cite{guerini2017operational, oszmaniec2017simulating, filippov2018simulability}). Essentially this is obtained by combining the two different classical simulations from Sec.\ \ref{sec:cl-pre-post}.

\begin{figure}[htb]
    \centering
    \includegraphics{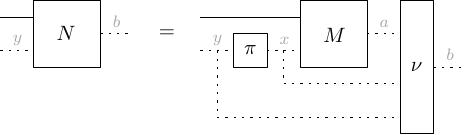}
    \caption{Classical simulation of multimeters.} 
    \label{fig:classical-simulation}
\end{figure}

\begin{defi}
   Let $N=\{N_{\cdot |y}\}_{y \in [r]} \subset \qm{l}{d}$ be a multimeter of $r$ POVMs each with $l$ outcomes on a $d$-dimensional Hilbert space. We say that $N$ can be \emph{classically simulated} (or is \emph{classically simulable}) with a multimeter $M = \{M_{\cdot|x}\}_{x \in [g]} \subset \qm{k}{d}$ of $g$ POVMs with $k$ outcomes (on the same Hilbert space) if there exist conditional probability distributions $\pi:= (\pi_{\cdot|y})_{y \in [r]}$ on $[g]$ and $ \nu := (\nu_{\cdot|a,x,y})_{a \in [k], x \in [g], y \in [r]} $ on $[l]$ such that 
\begin{equation}\label{eq:simulation}
    N_{b|y} = \sum_{x=1}^g \pi_{x|y} \sum_{a=1}^k \nu_{b|a,x,y} M_{a|x} 
\end{equation}
for all $b \in [l]$ and $y \in [r]$. 
\end{defi}

The operational interpretation of classical simulability is the following: We are conducting a physical experiment with a $d$-level quantum system where we can perform measurements with a multimeter $M$ with $g$ measurement settings. Given an input $y$ which corresponds to the label of the new measurement setting, with probability $\pi_{x|y}$ we choose the measurement setting $x$ and use the measurement  $M_{\cdot | x}$ to measure the system. After obtaining an outcome $a$ from the measurement of $M_{\cdot | x}$ instead of registering it we assign an outcome $b$ with probability $\nu_{b|a,x,y}$. Then, the resulting $r$ measurements (after multiple rounds of the experiment) are described by the multimeter $N$ in Eq.~\eqref{eq:simulation}; see Fig.\ \ref{fig:classical-simulation} for a graphical representation of this simulation scheme.

We see that classical simulation is a special case of a realization $(\Lambda, \nu)$ with a classical ancilla for the quantum supermap $\Psi: \M(\complex)_{kdg} \to \M(\complex)_{lnr}$. Namely, in Eq.\ \eqref{eq:ancilla-free} if we choose $s=1$, $n=d$ and take $\Lambda^*_{x|y} = \pi_{x|y} id_d$ for all $x \in [g]$ and $y \in [r]$ for some conditional probability distribution $\pi= (\pi_{\cdot|y})_{y \in [r]}$ on $[g]$, then 
\begin{equation}
\sum_{x=1}^g \sum_{a=1}^k \nu_{\cdot |a,y,x} \Lambda^*_{x|y}(M_{a|x}) =
 \sum_{x=1}^g \pi_{x|y} \sum_{a=1}^k \nu_{\cdot|a,x,y} M_{a|x} 
\end{equation}
for all $y \in [r]$ for all $M = \{M_{\cdot|x}\}_{x \in [g]} \subset \qm{k}{d}$. Thus, the classical simulation map is a particular instance of quantum superchannels between multimeters that admit a realization without even a classical ancilla. In particular, the realization only consists of classical pre- and postprocessing and both the original and the simulated multimeter act on the same-size quantum system.

The classical simulation scheme describes the construction of new measurements from existing ones by means of classical manipulations of the inputs and outputs of the measurement devices. Naturally, this is also linked to joint measurability since in the case of only one simulator, i.e., when $g=1$, the multimeter $M$ consists of only one POVM and the conditional probability distributions $\pi =(\pi_{\cdot|y})_{y \in [r]}$ are all trivial so that each POVM in $N$ can be postprocessed from the single POVM in $M$.

Furthermore, instead of just creating new multimeters from existing ones, one can also ask when a fixed multimeter can be simulated by some other multimeter with some desired properties. For example, one can ask when a multimeter can be simulated by a multimeter with a lesser number of measurements, or by a multimeter whose measurements have less number of outcomes, or by multimeters whose measurements are projective. Such topics have been explored in \cite{guerini2017operational, oszmaniec2017simulating, filippov2018simulability,Oszmaniec2019}.

\subsection{Compressibility} \label{sec:compressibility}
We continue by reviewing the results of the recent work \cite{IoannouSimulability2022} (see also \cite{Jones2007,Jokinen2023}), which we refer to as compressibility to distinguish the different notions of simulation of measurements. 

\begin{figure}[htb]
    \centering
    \includegraphics{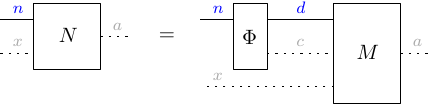}
    \caption{Compressibility of measurements. The quantum (solid) wires have dimensions indicated in blue.}
    \label{fig:compressibility}
\end{figure}

\begin{defi}
    Let $N:=\{N_{\cdot|x}\}_{x \in [g']} \subset \qm{k}{n}$ be a multimeter of $g'$ POVMs with $k$ outomes on an $n$-dimensional Hilbert space. We say that $N$ is \emph{$d$-compressible} if there exists a finite $C \in \nat$, a quantum instrument $\Phi \in \qi{C}{n}{d}$ and another multimeter $M=\{M_{\cdot|x,c}\}_{x \in [g'], c \in [C]} \subset \qm{k}{d}$ of $g' \cdot C$ POVMs with $k$ outcomes on a $d$-dimensional Hilbert space, such that 
\begin{equation}\label{eq:compression}
    N_{a|x} = \sum_{c =1}^C \Phi^\ast_c(M_{a|x,c})
\end{equation}
for all $a \in [k]$ and $x \in [g']$.
\end{defi}

Here the recipe of simulation is as follows:  a measurement with a label $x$ is obtained first by measuring the input state $\varrho \in \qs{n}$ by the instrument $\Phi$, obtaining an outcome $c$ and transforming the system into a conditional output state $\Phi_c(\varrho)$ which is now a (unnormalized) state of a $d$-dimensional system. This state is then measured by the POVM $M_{\cdot|x,c}$ from which an outcome $a $ is obtained and then reported as the final outcome of the transformed measurement with a label $x$. We depict this simulation scheme in Fig.\ \ref{fig:compressibility}. The terminology for compressibility comes from the case when $d<n$ so that one can simulate the measurements by performing some other measurements on a smaller quantum system. Note that \cite{IoannouSimulability2022} allows for non-finite $C$, which we exclude to avoid technical difficulties.

Similarly to classical simulability, in our framework we get compressibility as a special case of a realization $(\Lambda, \nu)$ with a classical ancilla for the quantum supermap $\Psi: \M(\complex)_{kdg} \to \M(\complex)_{lnr}$. Namely, in Eq.\ \eqref{eq:ancilla-free} we may choose $s=1$, $g=g' \cdot C$, $l=k$ and $r=g'$, and set $\nu_{a|b,x',c,x} = \mathbf{1}_{a=b}$ for all $a,b \in [k]$ and $x,x' \in [g']$, $c \in [C]$ and $\Lambda^*_{x',c|x} = \mathbf{1}_{x=x'} \Phi^*_c$ for all $x,x' \in [g']$ and $c \in [C]$ for some instrument $\Phi \in \qi{C}{n}{d}$ so that
\begin{equation}
    \sum_{x'=1}^{g'} \sum_{c =1}^C\sum_{a=1}^k \nu_{\cdot |a,x',c,x} \Lambda^*_{x',c|x}(M_{a|x', c}) =  \sum_{c=1} ^C  \Phi^*_c(M_{\cdot |x, c}) 
\end{equation}
for all $x \in [g']$ and for all multimeters $M=\{M_{\cdot|x,c}\}_{x \in [g'], c \in [C]} \subset \qm{k}{d}$. Thus, also compressibility is indeed a particular instance of a quantum superchannel between multimeters with a realization with a classical ancilla.

It is easy to see that POVMs are $1$-compressible if and only if they are compatible, because in this case the $M_{a|x,c}$ are just conditional probabilities and the effects $E_c := \Phi^\ast_c(1) \in \mathcal M(\complex)^{\mathrm{sa}}_n$ form a joint POVM $E$ for the multimeter $N$. In addition to generalizing compatibility, compressibility (also called high-dimensional simulability) is shown to be equivalent to high-dimensional steering \cite{Jones2023}.

\subsection{Compatibility-preserving simulations}\label{sec:comp-preserving}

In the previous cases one could see both simulations as generalizations of compatibility of multimeters so that compatibility emerges only as a special instant of the simulation. In the work \cite{buscemi2020complete}, however, the authors discuss a class of superchannels $\Psi$ that preserve the property of compatibility of multimeters. Their motivation for introducing such superchannels is to use them to build a resource theory of quantum incompatibility where such maps would act as free operations between the objects of the resource theory. In our setting, we can rephrase their definition of ``programmable measurement device  (PMD) processing''  as follows:

\begin{defi}
    Let $N:=\{N_{\cdot|y}\}_{y \in [r]} \subset \qm{l}{n}$ be a multimeter of $r$ POVMs with $l$ outomes on an $n$-dimensional Hilbert space. We say that $N$ can be \emph{compatibility-preservingly simulated} by a multimeter $M=\{M_{\cdot|x}\}_{x \in [g]}\subset \qm{k}{d}$ of $g$ POVMs with $k$ outcomes on a $d$-dimensional Hilbert space if there exists $K,L \in \nat$, a probability distribution $p$ on $[K]$, a set of quantum instruments $\Gamma := \{\Gamma_{\cdot|\kappa}\}_{\kappa \in [K]}\subset \qi{L}{n}{d}$ and conditional probability distributions $\pi := (\pi_{\cdot|y,\lambda,\kappa})_{y \in [r], \lambda \in [L], \kappa\in [K]}$ on $[g]$ and $\nu := (\nu_{\cdot|a,x,y, \lambda,\kappa})_{a \in [k], x \in [g], y \in [r], \lambda \in [L], \kappa \in [K]}$ on $[l]$ such that 
    \begin{equation}
        N_{b|y} =  \sum_{\kappa=1}^K \sum_{x=1}^g \sum_{a=1}^k  \sum_{\lambda=1}^L p_\kappa \nu_{b|a,x,y,\lambda,\kappa} \pi_{x|y, \lambda,\kappa} \Gamma^*_{\lambda|\kappa}(M_{a|x}) \label{eq:comp-preserving}
    \end{equation}
    for all $b \in [l]$ and $y \in [r]$.
\end{defi}

The interpretation is as follows: given an input state $\varrho \in \qs{n}$ and a classical input $y$ for the label of the new measurement, we choose an instrument $\Gamma_{\cdot|\kappa}$ according to probability $p_\kappa$ and measure the input state with it. The measurement leads to an outcome $\lambda$ and the state is transformed to a conditional output state $\Gamma_{\lambda|\kappa}(\varrho)$. Now, given $y$, $\kappa$ and $\lambda$ we choose label $x$ of the simulator POVM with probability $\pi_{x|y, \lambda, \kappa}$ and measure the (conditional) state $\Gamma_{\lambda|\kappa}(\varrho)$ with the POVM $M_{\cdot|x}$. After the measurement we obtain an outcome $a$ which we finally postprocess  (by taking also into account the classical information $x,y,\lambda,\kappa$) into an outcome $b$ with probability $\nu_{b|a,x,y,\lambda,\kappa}$ and report it as the final outcome of the new measurement $y$. We depict this process in Fig.\ \ref{fig:compatibility-preserving}. 
The motivation  behind the term ``compatibility-preserving simulation'' comes from the fact that if the multimeter $M$ is compatible then the resulting multimeter $N$ in Eq.\ \eqref{eq:comp-preserving} is compatible as well \cite{buscemi2020complete}. However, it is currently an open question whether all transformations between multimeters that preserve the compatibility of the multimeters are of this form.

\begin{figure}[htb]
    \centering
    \includegraphics{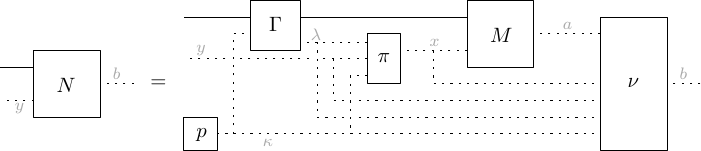}
    \caption{A compatibility-preserving transformation of a multimeter $M$; compare with \eqref{eq:comp-preserving}.}
    \label{fig:compatibility-preserving}
\end{figure}

Again we can show that the above simulation scheme can be presented as a special case of a realization with a classical ancilla of a the quantum supermap $\Psi: \M(\complex)_{kdh} \to \M(\complex)_{lnr}$ that takes multimeters to multimeters. We first note that in Fig.\ \ref{fig:compatibility-preserving} there are two classical ancillas connecting the preprocessing side and the postprocessing side. Hence, in Eq.\ \eqref{eq:ancilla-free} we set $s= K \cdot L$. Then we simply choose $\Lambda^*_{x,\lambda,\kappa|y} =p_\kappa \pi_{x|y,\lambda,\kappa} \Gamma^*_{\lambda|\kappa}$ for all $x \in [g]$, $y \in [r]$, $\lambda \in [L]$ and $\kappa\in [K]$ for some probability distribution $p$ on $[K]$, some set of quantum instruments $\Gamma := \{\Gamma_{\cdot|\kappa}\}_{\kappa \in [K]}\subset \qi{L}{n}{d}$ and some conditional probability distribution $\pi := (\pi_{\cdot|y,\lambda,\kappa})_{y \in [r], \lambda \in [L], \kappa\in [K]}$ on $[g]$. The transformed POVMs then take the form
\begin{align}
    \sum_{\kappa=1}^K \sum_{\lambda=1}^L \sum_{x=1}^g \sum_{a=1}^k  \nu_{\cdot |a,x,y, \lambda,\kappa} \Lambda^*_{x, \lambda, \kappa|y}(M_{a|x})  = \sum_{\kappa=1}^K \sum_{\lambda=1}^L \sum_{x=1}^g \sum_{a=1}^k  p_\kappa \nu_{\cdot |a,x,y, \lambda,\kappa}  \pi_{x|y,\lambda,\kappa} \Gamma^*_{\lambda|\kappa}(M_{a|x})
\end{align}
for all $y \in [r]$. Thus, also compatibility-preserving simulation is a particular instance of a quantum superchannel between multimeters that has a realization with a classical ancilla.

\section{A new notion of simulation of multimeters} \label{sec:multimeter-simulation}

\subsection{Not all transformations are simulations}
A simulation of multimeters is a process that takes an existing multimeter and transforms it to another multimeter. However, not all possible transformations can be considered to capture the essence of what would be considered a simulation. For example, one could state that in order for a transformation between multimeters to be truly considered a simulation, the simulation process should minimally involve using at least some parts of the original multimeter. Or by considering simulation as a resource theory, one could argue that not all transformations between multimeters can be considered as free operations since otherwise one could turn any object into another object freely so that there wouldn't be any resources to begin with. To see which types of transformations should be left out of simulations, we start by presenting the following example:

\begin{ex}
    Let us revisit Example \ref{ex:t-a-p-not-unique} and consider a process where we take the original multimeter, discard it and replace it with some other fixed multimeter. More precisely, such transformations are superchannels $\Psi: \M(\complex)_{kdg} \to \M(\complex)_{lnr}$ between multimeters such that $\Psi(J_{\Phi_M}) = J_{\Phi_N}$ for some fixed $N = \{N_{\cdot|y}\}_{y\in [r]} \subset \qm{l}{n}$ for all input multimeters $M=\{M_{\cdot|x}\}_{x \in [g]} \subset \qm{k}{d}$. In particular, even a trivial multimeter, i.e., a multimeter that consists of trivial POVMs $M = \{M_{\cdot|x}\}_{x \in [g]}$, where $M_{a|x} = p_{a|x} \id$ for all $a \in [k]$ for some conditional probability distribution $p = (p_{\cdot|x})_{x \in [g]}$ on $[k]$, is mapped to the fixed multimeter $N$.
\end{ex}

The two main points that we can infer from the previous examples are the following: First, since any multimeter $M$ is mapped to a fixed multimeter $N$, the simulation process corresponding to the previous map is not using any part of the simulator multimeter to perform the simulation. Second, if $N$ includes a nontrivial POVM, then trivial multimeters can simulate nontrivial ones. This suggests that it’s the simulation process, rather than the multimeter we aim to use, that extracts information from the quantum state. If we put ourselves in the position of an experimenter who has a multimeter at her disposal and would like to know what other measurements she can perform without having to change her experiment completely, these two things are not desirable in our opinion: If the experimenter wants the simulation process to engage at least part of the simulating device, rather than disregarding it, then transformations that map any multimeter to a fixed multimeter should be avoided. Furthermore, because trivial multimeters discard the quantum state without measuring it, if a transformation converts trivial multimeters into nontrivial ones, it indicates that an additional device is needed to extract information from the quantum state.

Hence, for a transformation between multimeters $\Psi: \M(\complex)_{kdg} \to \M(\complex)_{lnr}$ we list the following properties. We say that $\Psi$ is
\begin{enumerate}
    \item \emph{trash-and-prepare} if for all multimeters $M$ we have that $\Psi(J_{\Phi_M}) = J_{\Phi_N}$ for some fixed multimeter $N$,
    \item \emph{triviality-preserving} if whenever $M$ consists of only trivial POVMs, then $\Psi(J_{\Phi_M})$ corresponds to a multimeter that consists of trivial POVMs.
\end{enumerate}

An important thing to notice is that actually imposing that a map is triviality-preserving rules out most trash-and-prepare maps: Namely, if the map is triviality-preserving and trash-and-prepare, then the fixed multimeter $N$ to which it maps every multimeter $M$ must be a trivial multimeter. Thus, the only triviality-preserving trash-and-prepare maps are maps that take any multimeter to a fixed trivial multimeter. We note that operationally these type of maps may be considered simulations since trivial multimeters can be considered as free objects since they are defined only by the classical conditional probability distributions which can be thought of as an experimenter choosing to discard the quantum system and outputting a random number instead. This is something the experimenter can always do without having to change the measurement setup, i.e., the multimeter, even though it is a very poor use of the original simulator. Following this observation, we define what we consider in this article to be a simulation of multimeters:

\begin{defi}[Simulation of multimeters]\label{def:min-def-sim}
     A \emph{simulation of multimeters} is a transformation between multimeters, i.e., a quantum superchannel between multimeters, that is triviality-preserving.
\end{defi}

We note that, in general, the decision of which transformations to exclude depends on the intended application, allowing for the consideration of various different simulation schemes. However, in this work we focus on the scenario that we already describe earlier where an experimenter has a fixed multimeter to extract information from a quantum system and she is considering how other possible measurements could be implemented by using her device. Naturally, also in our set-up one can argue that maybe some other types of maps should be excluded from the above definition of simulability (such as maps that are not compatibility-preserving as in \cite{buscemi2020complete}). However, in this work we will focus on the above definition, which is a minimal definition for us, leaving us with a maximal set of simulation maps and present examples falling into this category of maps. Furthermore, next we will focus only on maps that admit a realization with a classical ancilla (or in the case when they are completely ancilla-free) since, as shown in the previous section, all the previously defined notions of simulations are of this type as well and we want to explore how these previous notions fall into our framework of simulation. We leave the treatment of the maps with a quantum ancilla for future work.

\subsection{Triviality-preserving maps}

Motivated by our minimal definition of simulability of multimeters (Def.\  \ref{def:min-def-sim}) we will next explore the structure of the realizations of a triviality-preserving map. In particular, in the case when a multimeter transformation admits an ancilla-free realization, i.e. $s=1$, we can show the following characterization result for the map preserving triviality.

\begin{thm}\label{thm:ancilla-free-tp-maps}
    Let $\Psi: \M(\complex)_{kdg} \to \M(\complex)_{lnr}$ be a quantum superchannel between multimeters that admits an ancilla-free realization (neither quantum nor classical ancilla). The following assertions are equivalent:
    \begin{enumerate}
        \item The transformation $\Psi$ is triviality-preserving.
        \item $\Psi$ admits an ancilla-free realization $(\Lambda, \nu)$ with the property that the multi-instrument $\Lambda$ is \emph{partially normalized} on the quantum system (see also Fig.~\ref{fig:Lambda-partially-normalized-quantum}): there exists a conditional probability distribution $\pi=(\pi_{\cdot|y})_{y \in [r]}$ on $[g]$ such that
        $$\Lambda^*_{x | y} (\id) = \pi_{x|y}\id$$
        for all $x \in [g]$ and $y \in [r]$.

\begin{figure}[htb]
    \centering
    \includegraphics{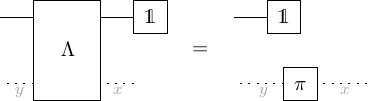}
    \caption{A multi-instrument $\Lambda$ that is partially normalized on the quantum system (continuous line).}
    \label{fig:Lambda-partially-normalized-quantum}
\end{figure}
        
        \item $\Psi$ admits an ancilla-free realization $(\Lambda, \nu)$ with the property that the multi-instrument $\Lambda$ factorizes as follows (see also Fig.~\ref{fig:Lambda-factorises-triviality-preserving}): there exists a conditional probability distribution $\pi = (\pi_{\cdot|y})_{y \in [r]}$ on $[g]$ and a family of $g\cdot r$ \emph{quantum channels} $\{\Phi_{x,y}\}_{x \in [g], y\in [r]} \subset \qc{n}{d}$ such that
        $$\Lambda_{x | y} = \pi_{x|y} \Phi_{x,y}$$
        for all $x \in [g]$ and $y \in [r]$.

\begin{figure}[htb]
    \centering
    \includegraphics{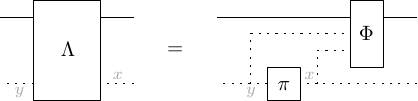}
    \caption{A multi-instrument $\Lambda$ that factorises and induces a triviality-preserving multimeter transformation $\Psi$.}
    \label{fig:Lambda-factorises-triviality-preserving}
\end{figure}
    \end{enumerate}
\end{thm}
\begin{proof}
    We start by showing (3) $\implies$ (1). Consider $\Psi$ having an ancilla-free realization $(\Lambda,\nu)$ such that $\Lambda_{x|y} = \pi_{x|y} \Phi_{x,y}$ so that also $\Lambda^*_{x|y} = \pi_{x|y} \Phi^*_{x,y}$ for all $x \in [g]$ and $y \in [r]$ for some set of quantum channels $\{\Phi_{x,y}\}_{x \in [g], y \in [r]} \subset \qc{n}{d}$ and some conditional probability distribution $\pi = (\pi_{\cdot|y})_{y \in [r]}$ on $[g]$. Let $M = \{M_{\cdot|x}\}_{x \in [g]} \subset \qm{k}{d}$ be a trivial multimeter, i.e., there exists some conditional probability distribution $p=(p_{\cdot|x})_{x \in [g]}$ on $[k]$ such that $M_{a|x} = p_{a|x} \id$ for all $a \in [k]$ and $x \in [g]$. We now have for all $b \in [l]$ and $y \in [r]$ that
    \begin{align*}
         &\sum_{x=1}^g \sum_{a=1}^k \nu_{b |a,x,y} \Lambda^*_{x|y}(M_{a|x}) = \sum_{x=1}^g \sum_{a=1}^k \nu_{b |a,x,y} \pi_{x|y} p_{a|x} \Phi^*_{x,y}(\id) \\& = \left(\sum_{x=1}^g \sum_{a=1}^k \nu_{b |a,x,y} \pi_{x|y} p_{a|x} \right) \id =: q_{b|y} \id,
    \end{align*}
    where the last equality follows from unitality of the channels $\{\Phi^*_{x,y}\}_{x \in [g], y \in [r]}$. Since clearly $q=( q_{\cdot|y})_{y\in [r]}$ is a set of conditional probability distributions on $[l]$, it follows that $\Psi$ is triviality-preserving.

\smallskip

    Let us now prove (2) $\implies$ (3). Define for all $x\in[g]$ and $y \in [r]$:
    $$\Phi_{x,y}:= \begin{cases}
        \frac{1}{\pi_{x|y}} \Lambda_{x|y} &\qquad \text{ if }\pi_{x|y} \neq 0\\
        \Phi_0&\qquad \text{ if }\pi_{x|y} = 0,
    \end{cases}$$
where $\Phi_0$ is some fixed quantum channel (which does not play any role). Since $\Lambda$ is a multi-instrument, the maps $\Phi_{x,y}$ defined above are completely positive, and 
$$\forall x,y \qquad \Phi_{x,y}^*(\id) = \id,$$
proving that they are indeed quantum channels; this concludes the proof of (2) $\implies$ (3).
\smallskip

    We show now the last implication, (1) $\implies$ (2). Let $\Psi$ transform multimeters as in Eq.~\eqref{eq:ancilla-free}, with $s=1$ (no classical ancilla $\lambda$). In order for this to be triviality-preserving for any trivial multimeter $M=\{M_{\cdot|x}\}_{y\in [r]} = \{p_{\cdot|x} \id\}_{x\in [g]} \subset \qm{k}{d}$ it should result in some other trivial multimeter $N^M = \{N^M_{\cdot|y}\}_{y\in [r]}=\{q^p_{\cdot|y}\id\}_{y\in [r]} \subset \qm{l}{n}$ so that 
    \begin{equation}
    q^p_{b|y}\id = N^M_{b|y} = \sum_{x=1}^g \sum_{a=1}^k \nu_{b |a,x,y} \Lambda^*_{x|y}(M_{a|x}) =\sum_{x=1}^g \sum_{a=1}^k \nu_{b |a,x,y} p_{a|x}  \Lambda^*_{x|y}(\id)
    \end{equation}
    for all $ b\in [l]$ and $y \in [r]$. Our goal is to show that, for a possibly different realization $(\tilde \Lambda, \tilde \nu)$ of $\Psi$, $\tilde \Lambda^*_{x|y}(\id) \sim \id$ for all $x,y$.

    Note that we can reason individually for every  setting $y \in [r]$ of the resulting multimeter, hence we shall omit the variable $y$ in the rest of this proof, for the sake of simplicity. Moreover, we only need to check the previous equation for \emph{extremal} trivial multimeter. Those multimeters are parametrized by functions $\alpha: [g] \to [k]$ by 
    $$p^{(\alpha)}_{a|x} = \mathbf{1}_{a = \alpha(x)} \qquad \forall x \in [g],$$
    where $\mathbf{1}$ is the indicator function.
    The condition above reads in this case: 
    \begin{equation}\label{eq:Lambda-id-alpha}
        \forall b \in [l]: \qquad \sum_{x=1}^g  \nu_{b |\alpha(x),x}  \Lambda^*_x(\id) = q^{(\alpha)}_b \id.
    \end{equation}

    We shall partition the set $[g]$ in two subsets, depending on the behavior of the conditional probabilities $\nu$ appearing in the given realization $(\Lambda, \nu)$ of $\Psi$: $[g] = X_\neq \sqcup X_=$ with
    \begin{align*}
        X_\neq &:= \{ x \in [g] \, : \, \exists b \in [l], \, a_1,a_2 \in [k] \, \text{ s.t. } \, \nu_{b|a_1, x} \neq \nu_{b|a_2, x} \}\\
        X_= &:= \{ x \in [g] \, : \, \forall b \in [l], \, \text{the function } a \mapsto \nu_{b|a, x} \text{ is constant} \}.
    \end{align*}

    We shall now show that for all $x \in X_\neq$, $\Lambda_x(\id) \sim \id$. To this end, fix $x_0 \in X_\neq$ and $b_0 \in [l]$, $a_1, a_2 \in [k]$ such that $\nu_{b_0|a_1, x_0} \neq \nu_{b_0|a_2, x_0}$. Choose a function $\alpha_1 : [g] \to [k]$ such that $\alpha_1(x_0) = a_1$ and define $\alpha_2 : [g] \to [k]$ by
    $$\alpha_2(x) = \begin{cases}
        \alpha_1(x) &\qquad \text{ if } x \neq x_0\\
        a_2 &\qquad \text{ if } x = x_0.
    \end{cases}$$
    With these choices of $\alpha_{1,2}$, taking the difference of Eq.~\eqref{eq:Lambda-id-alpha}, we obtain
    $$(\nu_{b_0|a_1, x_0} - \nu_{b_0|a_2, x_0}) \Lambda_{x_0}^*(\id) = (q^{(\alpha_1)}_{b_0} - q^{(\alpha_2)}_{b_0} ) \id,$$
    which allows us to conclude that $\Lambda_{x_0}^*(\id) \sim \id$, as claimed. 

    Let us now consider the case of indices $x \in X_=$. Since for such indices we cannot conclude as before, we shall \emph{construct another representation} $(\tilde \Lambda, \nu)$ of $\Psi$, with the property that $\tilde \Lambda_x^*(\id) \sim \id$ for all $x \in [g]$. We define, for all matrix $Z$,
    $$\tilde \Lambda^*_x(Z) := \begin{cases}
        \Lambda^*_x(Z) &\qquad \text{ if } x \in X_{\neq}\\
        \frac{\Tr{\Lambda^*_x(Z)}}{d} \id &\qquad \text{ if } x \in X_=.
    \end{cases}$$
    Since $\{ \Lambda^*_x \}_{x \in [g]}$ were a family of completely positive maps summing up to a unital CP map, the same holds for $\{\tilde \Lambda^*_x \}_{x \in [g]}$. We need to show that $(\tilde \Lambda, \nu)$ is indeed a representation of the map $\Psi$. For a given multimeter $M = \{M_{\cdot|x}\}_{x \in [g]} \subset \qm{k}{d}$, the transformed multimeters take the form 
    \begin{align*}
            \sum_{x=1}^g \sum_{a=1}^k \nu_{b |a,x} \Lambda^*_{x}(M_{a|x}) &= \sum_{x \in X_\neq} \sum_{a=1}^k \nu_{b |a,x} \Lambda^*_{x}(M_{a|x}) + \sum_{x \in X_=} \sum_{a=1}^k \underbrace{\nu_{b |a,x}}_{=:\nu_{b|x}} \Lambda^*_{x}(M_{a|x}) \\
            &= \sum_{x \in X_\neq} \sum_{a=1}^k \nu_{b |a,x} \tilde  \Lambda^*_{x}(M_{a|x}) + \sum_{x \in X_=} \nu_{b|x} \sum_{a=1}^k  \Lambda^*_{x}(M_{a|x}) \\
            &= \sum_{x \in X_\neq} \sum_{a=1}^k \nu_{b |a,x} \tilde  \Lambda^*_{x}(M_{a|x}) + \sum_{x \in X_=} \nu_{b|x}   \Lambda^*_{x}(\id).
    \end{align*}
    Hence, in order to conclude, we need to show that for all $b \in [l]$, 
    $$ \sum_{x \in X_=} \nu_{b|x}   \Lambda^*_{x}(\id) =  \sum_{x \in X_=} \nu_{b|x}   \tilde \Lambda^*_{x}(\id).$$
    To this end, note that Eq.~\eqref{eq:Lambda-id-alpha} implies that, for all functions $ \alpha$ and $b \in [l]$, we have that
    $$q^{(\alpha)}_b \id = \sum_{x=1}^g  \nu_{b |\alpha(x),x}  \Lambda^*_x(\id) = \sum_{x \in X_\neq}  \nu_{b |\alpha(x),x}  \pi_x \id + \sum_{x\in X_=}  \nu_{b|x}  \Lambda^*_x(\id).$$
    In particular, we obtain that  
    \begin{equation}\label{eq:q-b-=}
        \sum_{x\in X_=}  \nu_{b|x}  \Lambda^*_x(\id) = q^{(=)}_b \id
    \end{equation}
    for all $b\in[l]$ for some non-negative scalar $q^{(=)}_b$, independent of $\alpha$. Taking the trace of this expression yields
    $$q^{(=)}_b = \sum_{x\in X_=}  \nu_{b|x} \frac{\Tr{\Lambda^*_x(\id)}}{d}.$$
    Plugging this value back into Eq.~\eqref{eq:q-b-=} we obtain
    $$\sum_{x\in X_=}  \nu_{b|x}  \Lambda^*_x(\id) = \sum_{x\in X_=}  \nu_{b|x} \frac{\Tr{\Lambda^*_x(\id)}}{d} \id = \sum_{x\in X_=}  \nu_{b|x} \tilde \Lambda^*_x(\id),$$
    which was our goal.
\end{proof}

What our result thus shows is that a tranformation $\Psi$ between multimeters that admits an ancilla-free realization $(\Lambda,\nu)$ is triviality-preserving if and only if there is a (possibly different) ancilla-free realization $(\tilde{\Lambda}, \nu)$ where the preprocessing part $\tilde{\Lambda}$ factorizes into just probabilistically applying some set of channels instead of some general instruments. This is exactly what is demonstrated in Fig.\ \ref{fig:Lambda-factorises-triviality-preserving}. This reflects our original motivation of considering triviality preserving simulations in the first place: Indeed, in the case of no ancilla, the only way a transformation between multimeters can be triviality-preserving is when the preprocessing part does not extract information from the quantum state but rather just probabilistically transforms the state. It is worth noting that our proof is constructive so that given the original realization $(\Lambda,\nu)$ the proof can be used to find the other realization $(\tilde{\Lambda}, \nu)$ that satisfies the conditions (2) and (3).

We also note that even in the case when the original realization has a classical ancilla, the conditions (2) and (3) (with added classical index $\lambda$ as an outcome of the preprocessing $\Lambda$) imply that the map is triviality-preserving (this is essentially just the same calculation as in the first part of the proof). However, the precise necessary condition for the map being triviality-preserving in the case of classical ancilla is still an open question.

We can now straight-forwardly apply Thm.\ \ref{thm:ancilla-free-tp-maps} to the classical simulability map, the compressibility map and the compatibility-preserving map.

\begin{cor}[Classical simulation]\label{cor:cs-tp}
    The classical simulation map defined in Eq.\ \eqref{eq:simulation} is always triviality-preserving.
\end{cor}
\begin{proof}
    The explicit realization $(\Lambda, \nu)$ that we give in Sec.\ \ref{sec:classical-simulation} is defined by setting $n=d$ and $\Lambda_{x|y} = \pi_{x|y} id_d$ for all $x \in [g]$ and $y \in [r]$, where $id_d$ is the identity map on $\complex^d$. Clearly it is of the form given in the condition (2) of Thm.\ \ref{thm:ancilla-free-tp-maps}.
\end{proof}

\begin{cor}[Compressibility]\label{cor:c-tp}
    The compression map defined in Eq.\ \eqref{eq:compression} is triviality-preserving if and only if the compressing instrument $\Phi  \in \qi{C}{n}{d}$ is of the form $\Phi_c = \pi_{c} \Omega_{c}$ for all $c \in C$ for some probability distribution $\pi$ on $ [C]$ and some set of channels $\{\Omega_{c}\}_{c \in [C]} \subset \qc{n}{d}$.
\end{cor}
\begin{proof}
    As explained in Sec.\ \ref{sec:compressibility}, the compressibility map given by Eq.\ \eqref{eq:compression} admits an ancilla-free realization $(\Lambda,\nu)$, where $g=g' \cdot C$, $l=k$ and $r=g'$, $\nu_{a|a',x',c,x} = \mathbf{1}_{a=a'}$ for all $a,a' \in [k]$ and $x,x' \in [g']$, $c \in [C]$ and $\Lambda^*_{x',c|x} = \mathbf{1}_{x=x'} \Phi^*_c$ for all $x,x' \in [g']$ and $c \in [C]$ for some instrument $\Phi \in \qi{C}{n}{d}$. By applying the latter part of the proof of Thm.\ \ref{thm:ancilla-free-tp-maps} in this case, we see that since for all $x,x' \in [g]$ and $c \in [C]$ there exists $a,a',a'' \in [k]$ such that $\nu_{a|a',x',c,x} \neq \nu_{a|a'',x',c,x}$, we have that $(x',c) \in X_{\neq}$ for all $x' \in [g]$ and $c \in [C]$ so that we actually can take $\tilde{\Lambda} = \Lambda$ in the proof. This means that we can apply the necessary and sufficient condition (3) (or (2)) directly to the current realization $(\Lambda, \nu)$.
    
    Thus, by Thm.\ \ref{thm:ancilla-free-tp-maps} the compressibility map is triviality-preserving if and only if 
    $$\mathbf{1}_{x=x'} \Phi^*_c=\Lambda^*_{x',c|x} = \tilde{\pi}_{x',c|x} \tilde{\Omega}^*_{x',c,x}$$ 
    for all $c \in [C]$ and $x,x' \in [g]$ for some conditional probability distribution $\tilde{\pi}=(\tilde{\pi}_{\cdot,\cdot|x})_{x \in [g]}$ on $[g] \times [C]$ and some set of channels $\{\tilde{\Omega}_{x',c,x}\}_{c \in [C], x,x' \in [g]} \subset \qc{n}{d}$. It follows that \begin{equation}\label{eq:compressibility-tp}
        \Phi^*_c = \sum_{x'} \tilde{\pi}_{x',c|x} \tilde{\Omega}_{x',c,x} 
    \end{equation}
    and furthermore that $\Phi^*_c(\id) = \sum_{x'} \tilde{\pi}_{x',c|x} \id$ for all $c \in [C]$ and $x \in [g]$. Let us define a probability distribution $\pi$ on $[C]$ by setting $\pi_c := \sum_{x'} \tilde{\pi}_{x',c|x}$ for all $c \in [C]$ which we note that is now independent of $x \in [g]$. We note that $\pi_c \neq  0$ if and only if $\Phi^*_c(\id)\neq 0$ if and only if $\Phi^*_c \neq 0$. We can now define a set of CP maps $\{\Omega^*_c\}_{c \in [C]}$ by setting $\Omega^*_c = \Phi^*_c/\pi_c$ for all $\pi_c \neq 0$ and $\Omega^*_c= \Omega_0$ for all $\pi_c = 0$ for some fixed channel $\Omega_0 \in \qc{n}{d}$. From Eq.\ \eqref{eq:compressibility-tp} it follows that the maps are actually unital so that $\{\Omega_c\}_{c \in [C]} \subset \qc{n}{d}$. The claim follows.
\end{proof}

\begin{cor}[Compatibility-preserving]\label{cor:cp-tp}
    The compatibility-preserving map defined in Eq.\ \eqref{eq:comp-preserving} is triviality-preserving if the preprocessing instruments $\Gamma = \{\Gamma_{\cdot|\kappa}\}_{\kappa \in [K]} \in \qi{L}{n}{s}$ are of the form $\Gamma_{\lambda|\kappa} = \mu_{\lambda|\kappa} \Phi_{\lambda, \kappa}$ for all $\lambda \in [L]$ and $\kappa \in [K]$ for some conditional probability distribution $\mu=(\mu_{\cdot|\kappa})_{\kappa \in [K]}$ on $[L]$ and some set of channels $\{\Phi_{\lambda, \kappa}\}_{\lambda \in [L], \kappa \in [K]} \subset \qc{n}{d}$.
\end{cor}
\begin{proof}
    The claim follows from a straightforward calculation as in the beginning of the proof of Thm.\ \ref{thm:ancilla-free-tp-maps}.
\end{proof}
Since the realization of the compatibility-preserving map utilizes a classical ancilla, Thm.\ \ref{thm:ancilla-free-tp-maps} cannot be applied to see whether this condition is also a necessary one for the map to be triviality-preserving. We leave the necessary condition as an open question.

Finally, let us give a further example of an explicit map that is not triviality-preserving.

\begin{example}
Let us fix $r=1$ and $g=k=l=d=n= 2$ so that the map transforms two dichotomic qubit POVMs, say $M_{\cdot|0}$ and $M_{\cdot|1}$, to a single dichotomic qubit POVM, say $N$. Let the map $M \mapsto N$ be defined  as
$$N_b  = \begin{bmatrix}
    \braket{ 0 | M_{b|0} | 0} & 0 \\
    0 & \braket{ 1 | M_{b|1} | 1}
\end{bmatrix}$$
for all $b \in \{0,1\}$ for the computational basis $\{\ket{0}, \ket{1}\}$ of $\complex^2$.
This map admits a realization $(\Lambda, \nu)$ with $\nu_{b|a,x} = \mathbf{1}_{b=a}$ and 
\begin{align*}
    \Lambda_0(Z) = \Lambda^*_0(Z) &= \begin{bmatrix}
    \braket{ 0 | Z | 0} & 0 \\
    0 & 0
\end{bmatrix} \\
    \Lambda_1(Z) = \Lambda^*_1(Z) &= \begin{bmatrix}
    0 & 0 \\
    0 & \braket{ 1 | Z | 1}
\end{bmatrix}.
\end{align*}
Note that neither of these maps satisfy $\Lambda_x^*(\id) \sim \id$ and the map $\Psi$ is not triviality-preserving, since if we take $M=\{M_{\cdot|0}, M_{\cdot|1}\} = \{p_{\cdot|0}\id, p_{\cdot|1}\id\}$ for some probability distributions with $p_{0|0}=p$, $p_{1|0}=1-p$, $p_{0|1}=q$, $p_{1|1}=1-q$ for some $q, p \in [0,1]$, we see that $M$ is mapped to a POVM $N$ such that 
$$N_0 = \begin{bmatrix}
    p & 0 \\ 0 & q 
\end{bmatrix} , \quad
N_1 = \begin{bmatrix}
    1-p & 0 \\ 0 & 1-q 
\end{bmatrix}  $$
which is not trivial in general.
\end{example}

\subsection{Trash-and-prepare maps} 

Since triviality-preserving maps are mostly not trash-and-prepare (they are trash-and-prepare only in the case when a fixed trivial multimeter is prepared), in order to determine when a transformation is a simulation, i.e., triviality-preserving, it might be easier to first check that the map is not trash-and-prepare.

To start ruling out the trash-and-prepare maps from the general quantum superchannels between multimeters we make a simple observation. Let a superchannel $\Psi: \M(\complex)_{kdg} \to \M(\complex)_{lnr}$ admit a realization $(\complex^s,\Lambda,\nu)$ as in Thm.\ \ref{thm:multimeter-transformation-realization}. Let us assume that the POVMs in $B= \{B_{\cdot|a,x,y}\}_{a, \in [k],x \in [g], y \in [r]} \subset \qm{l}{s}$ are independent of the output $a \in [k]$ of the input multimeters $M$, i.e., $B_{\cdot|a,x,y} = B_{\cdot|a',x,y} =: B_{\cdot|x,y}$ for all $a,a' \in [k]$, $x \in [g]$ and $y \in [r]$. Now we see that the resulting POVMs are given by
\begin{equation*}
    \sum_{x=1}^g \sum_{a=1}^k \Lambda^*_{x|y}(M_{a|x} \otimes B_{\cdot |a,x,y}) = \sum_{x=1}^g \sum_{a=1}^k \Lambda^*_{x|y}(M_{a|x} \otimes B_{\cdot |x,y}) =  \sum_{x=1}^g \Lambda^*_{x|y}(\id_d \otimes B_{\cdot |x,y}) \in \qm{l}{n}
\end{equation*}
for all $y \in [r]$. Thus, in this case we see that the process of transforming the multimeter $M$ is just to ignore $M$ and prepare the resulting multimeter irrespective of $M$. This means that it is trash-and-prepare.

The above result is intuitive: in Fig.\ \ref{fig:multimeter-simulation} for the multimeter $B$ being independent of the outcome $a$ corresponds to having no classical wire connecting the input multimeter $M$ and the fixed multimeter $B$. If this is the case the outcome $a$ of the multimeter $M$ can be simply discarded and the multimeter $B$ is applied to the ancilla not affected by $a$ at all. Thus, in the end a fixed multimeter is applied irrespective of the input multimeter $M$.

We see that the above sufficient condition for a map being trash-and-prepare works even in the case of a quantum ancilla. However, in the case of a classical ancilla, it turns out that the trash-and-prepare maps are exactly of this type:

\begin{thm}\label{thm:ancilla-free-tap-maps}
    Let $\Psi: \M(\complex)_{kdg} \to \M(\complex)_{lnr}$ be a quantum superchannel that admits a realization $(s,\Lambda,\tilde \nu)$ with a classical ancilla. Then $\Psi$ is a trash-and-prepare map if and only if it admits a (possibly different) realization $(s,\Lambda,\nu)$ with a classical ancilla such that all the conditional probability distributions $\nu= \{\nu_{\cdot|a,x,y,\lambda}\}_{a \in [k],x \in [g],y\in[r],\lambda \in [s]}$ on $[l]$ are independent of $a \in [k]$. Furthermore, if $s=1$, then we can take $\nu = \tilde \nu$.
\end{thm}
\begin{proof}
    The sufficiency of the condition follows from the more general observation made before the statement of the theorem.
    
    On the other hand, if a quantum superchannel $\Psi$ with a realization $(s,\Lambda,\tilde \nu)$ with a classical ancilla transforms a set of POVMs in $\qm{k}{d}$ as in Eq.~\eqref{eq:ancilla-free}, then in order for this to be of the trash-and-prepare type, it should result in some fixed set of POVMs $N = \{N_{\cdot|y}\}_{y\in [r]} \subset \qm{l}{n}$ so that 
\begin{equation}\label{eq:no-ancilla-tap}
    N_{b|y} = \sum_{x=1}^g \sum_{a=1}^k \sum_{\lambda=1}^s \tilde \nu_{b |a,x,y, \lambda} \Lambda^*_{x,\lambda|y}(M_{a|x}) 
\end{equation}
for all $ b\in [l]$ and $y \in [r]$ for all sets of POVMs $M=\{M_{\cdot|x}\}_{x \in [g]} \subset \qm{k}{d}$. As this should hold for all $M$, if we take $M$ to consist of trivial POVMs, i.e.~$M_{a|x} = p_{a|x} \id_d$ for all $a \in [k]$ and $x \in [g]$ for some conditional probability distribution $p = (p_{\cdot|x})_{x \in [g]}$ on $[k]$ , then it follows that
\begin{equation}
    N_{b|y} = \left[ \sum_{x=1}^g \sum_{\lambda=1}^s \left(\sum_{a=1}^k p_{a|x} \tilde \nu_{b|a,x,y,\lambda} \right) \Lambda^*_{x,\lambda|y}  \right](\id_d) 
\end{equation}
for all $ b\in [l]$ and $y \in [r]$. Now if we fix $(a_1, \ldots, a_g)\in [k]^g$ and set $p_{a_x|x}=1$ for all $x \in [g]$, we see that 
\begin{equation}
    N_{b|y} = \left( \sum_{x=1}^g \sum_{\lambda=1}^s \tilde \nu_{b|a_x,x,y,\lambda} \Lambda^*_{x,\lambda|y}  \right)(\id_d) 
\end{equation}
for all $ b\in [l]$ and $y \in [r]$. Since he choice of $(a_1, \ldots, a_g)\in [k]^g$ was arbitrary, we have that  $N_{b|y} = \left( \sum_{x=1}^g \sum_{\lambda=1}^s \tilde \nu_{b|a,x,y,\lambda} \Lambda^*_{x,\lambda|y}  \right)(\id_d) = \left( \sum_{x=1}^g \sum_{\lambda=1}^s \tilde \nu_{b|a',x,y,\lambda} \Lambda^*_{x,\lambda|y} \right)(\id_d)$ for all $a,a' \in [k]$, $ b\in [l]$ and $y \in [r]$.

Let us again fix $(a_1, \ldots, a_g) \in [k]^g$ and furthermore let us fix also $x' \in [g]$ and $a'_{x'} \in [k]$ such that $a'_{x'} \neq a_{x'}$ and define POVMs $M=\{M_{\cdot|x}\}_{x \in [g]} \subset \qm{k}{d}$ by setting $M_{a'_{x'}|x'} = B$ and $M_{a_{x'}|x'} = \id_d - B$ for some fixed effect $B \in \qe{d}$ and $M_{a_{x}|x} = \id_d$ for all $x \neq x'$. Inserting these POVMs in Eq.~\eqref{eq:no-ancilla-tap}, we see that
\begin{align*}
    N_{b|y} &= \sum_{x=1}^g \sum_{a=1}^k \sum_{\lambda=1}^s \tilde \nu_{b |a,x,y,\lambda} \Lambda^*_{x,\lambda|y}(M_{a|x}) \\
    &= \sum_{x\neq x'} \sum_{\lambda=1}^s \tilde \nu_{b |a_x,x,y,\lambda} \Lambda^*_{x,\lambda|y}(\id_d) + \sum_{\lambda=1}^s \tilde \nu_{b|a'_{x'},x',y,\lambda} \Lambda^*_{x',\lambda|y}(B) +  \sum_{\lambda=1}^s \tilde \nu_{b|a_{x'},x',y,\lambda} \Lambda^*_{x',\lambda|y}(\id_d-B) \\
    &= \sum_{x=1}^g \sum_{\lambda=1}^s \tilde \nu_{b |a_x,x,y,\lambda} \Lambda^*_{x,\lambda|y}(\id_d) + \sum_{\lambda=1}^s \tilde \nu_{b|a'_{x'},x',y,\lambda} \Lambda^*_{x',\lambda|y}(B)-  \sum_{\lambda=1}^s \tilde \nu_{b|a_{x'},x',y,\lambda} \Lambda^*_{x',\lambda|y}(B) \\
    &=  N_{b|y} + \left(  \sum_{\lambda=1}^s \tilde \nu_{b|a'_{x'},x',y,\lambda} \Lambda^*_{x',\lambda|y}- \sum_{\lambda=1}^s \tilde \nu_{b|a_{x'},x',y,\lambda} \Lambda^*_{x',\lambda|y}\right)(B)
\end{align*}
for all $ b\in [l]$ and $y \in [r]$. Since $x',a_{x'},a'_{x'}$ were chosen arbitrarily and since the set of effects spans $\M(\complex)_d$, we must have that 
\begin{equation}\label{eq:t-a-p-cond}
    \sum_{\lambda=1}^s \tilde \nu_{b|a',x,y,\lambda} \Lambda^*_{x,\lambda|y}= \sum_{\lambda=1}^s \tilde \nu_{b|a,x,y,\lambda} \Lambda^*_{x,\lambda|y} 
\end{equation} 
for all $a,a' \in [k]$, $b \in[l]$, $x \in [g]$ and $y \in [r]$. Let us now define another set of conditional probability distributions $\nu = \{\nu_{\cdot|a,x,y,\lambda}\}_{a \in [k], x \in [g], y \in [r], \lambda \in [s]}$ on $[l]$ by setting 
$$
\nu_{b|a,x,y,\lambda} := \frac{1}{k} \sum_{a'=1}^k \tilde \nu_{b|a',x,y,\lambda}
$$ 
for all $a \in [k]$, $b \in [l]$, $x \in [g]$, $y \in [r]$ and $\lambda \in [s]$. By definition $\nu$ is now independent of the outcome $a \in [k]$, and we can show that $(s, \Lambda, \nu)$ is also a realization of $\Psi$ with a classical ancilla: Let us denote $\Omega_{b,x|y} := \sum_{\lambda} \tilde \nu_{b |a,x,y, \lambda} \Lambda^*_{x,\lambda|y}$ for all $b \in [l]$, $a \in [k]$, $x \in [g]$ and $y \in [r]$, which by Eq.\ \eqref{eq:t-a-p-cond} is indeed independent of the outcome $a \in [k]$. On the other hand, also 
\begin{align*}
&\Omega_{b,x|y} = \frac{1}{k} \sum_{a=1}^k \Omega_{b,x|y} = \frac{1}{k} \sum_{a=1}^k \sum_{\lambda=1}^s \tilde \nu_{b |a,x,y, \lambda} \Lambda^*_{x,\lambda|y} = \sum_{\lambda=1}^s \left( \frac{1}{k} \sum_{a=1}^k  \tilde \nu_{b |a,x,y, \lambda} \right) \Lambda^*_{x,\lambda|y} \\&= \sum_{\lambda=1}^s \nu_{b |a,x,y, \lambda} \Lambda^*_{x,\lambda|y} 
\end{align*}
for all $b \in [l]$, $a \in [k]$, $x \in [g]$ and $y \in [r]$. Now we see that
\begin{align*}
     \sum_{x=1}^g \sum_{a=1}^k \sum_{\lambda=1}^s \nu_{b |a,x,y,\lambda} \Lambda^*_{x,\lambda|y}(M_{a|x}) &=  \sum_{x=1}^g \sum_{a=1}^k \Omega_{b,x|y} (M_{a|x}) =  \sum_{x=1}^g \sum_{a=1}^k \sum_{\lambda=1}^s \tilde \nu_{b |a,x,y,\lambda} \Lambda^*_{x,\lambda|y}(M_{a|x})
\end{align*}
for all $ b\in [l]$, $y \in [r]$ and for all $M=\{M_{\cdot|x}\}_{x \in [g]} \subset \qm{k}{d}$. Hence, $(s,\Lambda,\nu)$ is also a realization of $\Psi$ and this completes the proof. In the case when $s=1$ we see that Eq.\ \eqref{eq:t-a-p-cond} already implies that $\tilde \nu$ is independent of the outcome $a \in [k]$ and we can choose $\nu = \tilde  \nu$.
\end{proof}

We can now again apply our result to the previously introduced simulation schemes.

\begin{cor}[Classical simulation]\label{cor:cs-tap}
     The classical simulation map defined in Eq.\ \eqref{eq:simulation} is trash-and-prepare if and only if the postprocessing $\nu= \{\nu_{\cdot|a,x,y}\}_{a \in [k],x \in [g],y\in[r]}$ is independent of the outcome $a \in [k]$. In this case the prepared multimeter is always trivial.
\end{cor}
\begin{proof}
    The first part of the statement follows from the case $s=1$ of the previous theorem. Now, if $\nu= \{\nu_{\cdot|a,x,y}\}_{a \in [k],x \in [g],y\in[r]}=:\{\nu_{\cdot|x,y}\}_{x \in [g],y\in[r]}$ is independent of the outcome $a \in [k]$, then the prepared multimeter is of the form
    \begin{equation*}
        \sum_{x=1}^g \pi_{x|y} \sum_{a=1}^k \nu_{\cdot |x,y} M_{a|x} = \sum_{x=1}^g \pi_{x|y} \nu_{\cdot |x,y} \id =: q_{\cdot |y} \id \in \qm{l}{d}
    \end{equation*}
    for all $y \in [r]$ and all multimeters $M= \{M_{\cdot|x}\}_{x \in [g]} \subset \qm{k}{d}$.
\end{proof}

\begin{cor}[Compressibility]\label{cor:c-tap}
     The compression map defined in Eq.\ \eqref{eq:compression} is never trash-and-prepare.
\end{cor}
\begin{proof}
    This follows from the case $s=1$ of the previous theorem when noting that the realization given in Sec.\ \ref{sec:compressibility} involves a postprocessing $\nu$ that is not independent of the outcome $a \in [k]$.
\end{proof}

\begin{cor}[Compatibility-preserving]\label{cor:cp-tap}
    The compatibility-preserving map defined in Eq.\ \eqref{eq:comp-preserving} with a realization $(L \cdot K, p \cdot \pi \cdot \Gamma, \tilde \nu)$ (as given in Sec.\ \ref{sec:comp-preserving}) is trash-and-prepare if and only if $(L \cdot K, p \cdot \pi \cdot \Gamma, \nu)$, where 
    $$
    \nu_{b|a,x,y,\lambda,\kappa} := \frac{1}{k} \sum_{a'=1}^k \tilde \nu_{b|a',x,y,\lambda,\kappa} \quad \forall  a \in [k], x \in [g], y \in [r], \lambda \in [L],\kappa \in [K],
    $$ 
    is also a realization.
\end{cor}
\begin{proof}
    This is just making the construction of $\nu$ obtained in the proof of Thm.\ \ref{thm:ancilla-free-tap-maps} explicit and rephrasing the original statement accordingly.
\end{proof}

Another curious application of Thm.\ \ref{thm:ancilla-free-tap-maps} is to show that in the absence of a quantum ancilla, there are maps that still require a classical ancilla in their realization.

\begin{example}\label{ex:cl-ancilla-needed}
    Consider the case of a trash-and-prepare map $\Psi$ transforming POVMs ($g=1$) to POVMs ($r=1$), preparing a non-trivial POVM $N$. We shall prove that such a map $\Psi$ cannot admit a realization without a quantum or a classical ancilla. If this were the case, so that $\Psi$ would have a realization $(\Lambda, \nu)$ with $s=1$. Then, Thm.\ \ref{thm:ancilla-free-tap-maps} would imply that $\nu = \{\nu_{\cdot|a}\}_{a \in [k]}$ is independent of $a \in [k]$ so that
    $$ N_b =  \sum_{a=1}^k \nu_{b|a} M_a = \nu_{b|a} \id_n$$
    for all $b \in [l]$, $a \in [k]$ and all POVMs $M$. This contradicts the fact that $N$ is non-trivial. 

    Of course, such a trash-and-prepare map can be realised with a \emph{classical} ancilla of size $s=l$, simply by defining
    $$\Lambda_z(\rho) = \Tr[\rho N_z]\sigma \qquad \forall z \in [l]$$
    and all input states $\rho$, for some fixed quantum state $\sigma$ and $\nu_{b|a,z} = \mathbf 1_{b=z}$ (see Fig.\ \ref{fig:classical-ancilla-needed}).

    \begin{figure}[htb]
        \centering
        \includegraphics{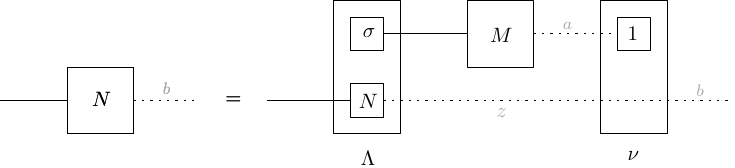}
        \caption{A trash-and-prepare POVM transformation that requires a classical ancilla.}
        \label{fig:classical-ancilla-needed}
    \end{figure}
\end{example}

\subsection{Comparing different simulations}

We will finish our investigation by considering the inclusions of the set of the previously considered maps.

    \begin{figure}[htb]
        \centering
        \includegraphics[scale=0.8]{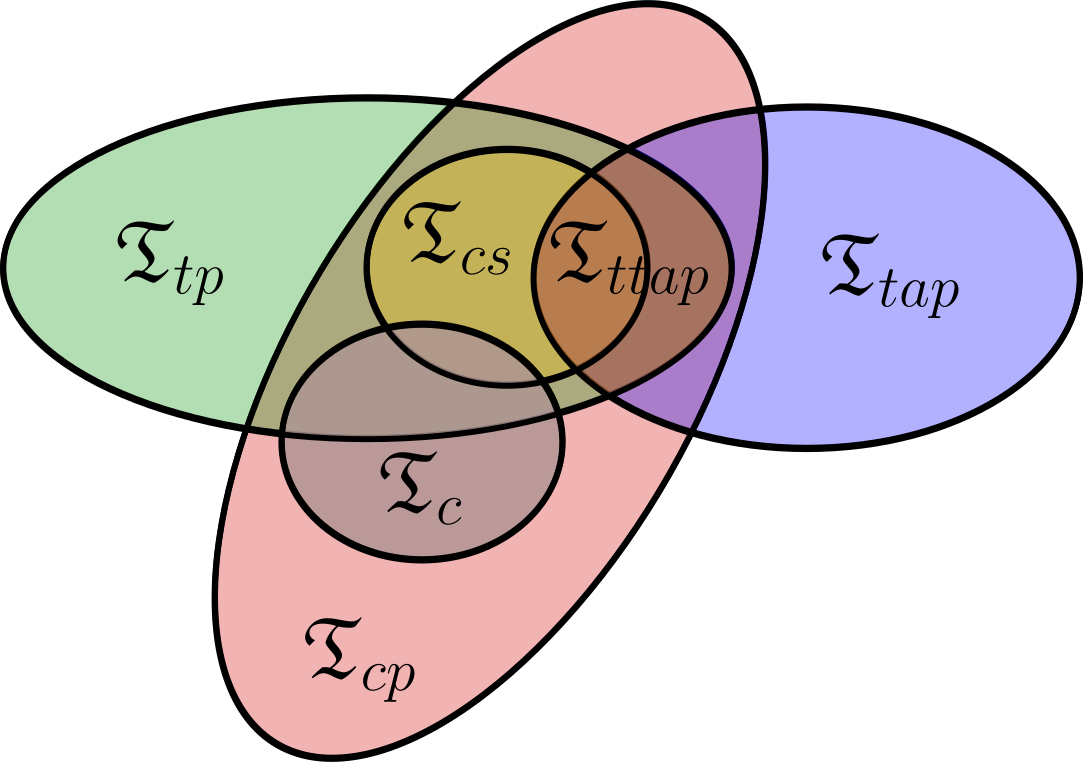}
        \caption{The inclusions of the sets of the investigated maps.}
        \label{fig:maps}
    \end{figure}

\begin{prop}
    Let us denote by $\mathfrak{T} :=\mathfrak{T}(\mathcal{MM}(g,k,d),\mathcal{MM}(r,l,n))$ the set of transformations $\Psi: \mathcal{M}_{gkd} \to \mathcal{M}_{rln}$ between multimeters, and let us consider the following sets of maps:
    \begin{align*}
        \mathfrak{T}_{tp} &:= \{\Psi \in \mathfrak{T} : \Psi \ \text{is  triviality-preserving} \}, \\
        \mathfrak{T}_{tap} &:= \{\Psi \in \mathfrak{T} : \Psi \ \text{is  trash-and-prepare} \}, \\
        \mathfrak{T}_{ttap} &:= \{\Psi \in \mathfrak{T}_{tap} : \Psi \ \text{prepares a fixed trivial multimeter} \}, \\
        \mathfrak{T}_{cs} &:= \{\Psi \in \mathfrak{T} : \Psi \ \text{is  a  classical  simulation}, \} \\
        \mathfrak{T}_{c} &:= \{\Psi \in \mathfrak{T} : \Psi \ \text{is  a  compression} \}, \\
        \mathfrak{T}_{cp} &:= \{\Psi \in \mathfrak{T} : \Psi \ \text{is  compatibility-preserving} \}.
    \end{align*}
    Then the inclusions presented in Fig.\ \ref{fig:maps} hold.
\end{prop}

\begin{proof}
    Let us start with the sets $\mathfrak{T}_{tap}$ and $\mathfrak{T}_{tp}$. The fact that $\mathfrak{T}_{ttap} = \mathfrak{T}_{tap} \cap \mathfrak{T}_{tp}$ is clear: As already stated before, if a map is triviality-preserving and trash-and-prepare, then the prepared fixed multimeter must be trivial. On the other hand, any trash-and-prepare map that prepares fixed trivial multimeters is triviality-preserving. 

    Let us now consider the set of classical simulations $\mathfrak{T}_{cs}$ and see how they relate to $\mathfrak{T}_{tap}$ and $\mathfrak{T}_{tp}$. The inclusion $\mathfrak{T}_{cs} \subset \mathfrak{T}_{tp}$ follows from Cor.\ \ref{cor:cs-tp} and by Cor.\ \ref{cor:cs-tap} we have that $\mathfrak{T}_{cs} \cap \mathfrak{T}_{tap} \subset \mathfrak{T}_{ttap}$. However, since maps in $\mathfrak{T}_{cs}$ cannot change the dimension of the quantum system, we cannot have the equality $\mathfrak{T}_{ttap}=\mathfrak{T}_{cs} \cap \mathfrak{T}_{tap}$ in general. But, on the other hand, if we have a map $\Psi: \mathcal{M}_{gkd} \to \mathcal{M}_{rln}$, $\Psi \in \mathfrak{T}_{ttap}$, such that $n=d$, and which prepares some fixed trivial multimeter $N = \{q_{\cdot|y}\id\}_{y \in [r]} \subset \qm{l}{d}$, then by defining the postprocessing $\nu= \{\nu_{\cdot|,a,x,y}\}_{a \in [k], x \in [g], y \in [r]}$ for the classical simulation as $\nu_{b|a,x,y} = q_{b|y}$ for all $b \in [l]$, $a \in [k]$, $x \in [g]$ and $y \in [r]$, it follows that
    \begin{equation*}
        \sum_{x=1}^g \pi_{x|y} \sum_{a=1}^k \nu_{\cdot |a,x,y} M_{a|x} = \sum_{x=1}^g \pi_{x|y} q_{\cdot |y} \id = q_{\cdot |y} \id 
    \end{equation*}
    for all $y \in [r]$ and all multimeters $M= \{M_{\cdot|x}\}_{x \in [g]} \subset \qm{k}{d}$, so that in this case $\Psi \in \mathfrak{T}_{cs}$.

    Let's now move on to the set $\mathfrak{T}_{c}$. The fact that $\mathfrak{T}_{c} \cap \mathfrak{T}_{cs} \neq \emptyset$ follows by considering a compression map with a compressing instrument $\Phi \in \qi{C}{d}{d}$ of the form $\Phi_c = \pi_c id_d$ for all $c \in [C]$. The "compressed" transformed multimeter is then of the form $\sum_c p_c M_{\cdot|x,c} $ for all $x \in [g]$ which is clearly just a mixture of the multimeter $M$ which is a special case of a classical simulation. We note that there are also compressions that are triviality-preserving which are not classical simulations so that $(\mathfrak{T}_c \cap \mathfrak{T}_{tp}) \setminus \mathfrak{T}_{cs} \neq \emptyset$: take now the compressing instrument as $\Phi_c = \pi_c \Phi_0$ for all $c \in [C]$ for some probability distrubution $\pi$ on $[C]$ and some fixed channel $\Phi_0 \in \qc{n}{d}$ for $n\neq d$. By Cor.\ \ref{cor:c-tp} the resulting compression map is then triviality-preserving but it is clearly not a classical simulation. Furthermore, there are also classical simulations that are not compressions (any classical simulation which non-trivially postprocesses the outcomes of the input multimeter to some different number of outcomes, i.e., $l \neq k$) so that also $\mathfrak{T}_{cs} \setminus \mathfrak{T}_{c} \neq \emptyset$. From Cor.\ \ref{cor:c-tp} one can see that there clearly are compression maps that are not triviality-preserving so that $\mathfrak{T}_{c} \setminus \mathfrak{T}_{tp} \neq \emptyset$. Finally, by Cor.\ \ref{cor:c-tap} we have that $\mathfrak{T}_{c} \cap \mathfrak{T}_{tap} = \emptyset$. 

    Lastly, let us focus on the compatibility-preserving maps $\mathfrak{T}_{cp}$. First, clearly any classical simulation can be obtained from a compatibility preserving map (Eq.\ \eqref{eq:comp-preserving}) by simply choosing $n=d$, $L=K=1$,  and by fixing the channel $\Gamma \in \qc{d}{d}$ as $\Gamma = id_d$. This shows that $\mathfrak{T}_{cs} \subset \mathfrak{T}_{cp}$. Also, any compression with a compressing instrument $\Phi \in \qi{C}{n}{d}$ can be obtained from the compatibility-preserving map by setting $g = g' \cdot C$, $K=1$, $L=C$, $l =k$, $r =g'$, and $\nu_{a|a',x',c',x,c} = \mathbf{1}_{a=a'}$, $\pi_{x',c'|x,c} = \mathbf{1}_{x=x'}\mathbf{1}_{c=c'}$ and $\Gamma_{c} = \Phi_{c}$ for all $a,a' \in [k]$, $x,x' \in [g]$ and $c,c' \in [C]$. Indeed, the resulting multimeter looks as follows:
    \begin{equation*}
         \sum_{x'=1}^g \sum_{c=1}^C \sum_{a=1}^k \sum_{c'=1}^C \nu_{a|a',x',c',x,c} \pi_{x',c'|x, c} \Gamma^*_{c}(M_{a'|x',c'}) =  \sum_{c=1}^C \Phi^*_c(M_{a|x,c}) 
    \end{equation*}
    for all $a \in [k]$ and $x \in [g]$. This show that $\mathfrak{T}_{c} \subset \mathfrak{T}_{cp}$. Furthermore, to see that also $\mathfrak{T}_{ttap} \subset \mathfrak{T}_{cp}$, we note that for any trivial multimeter $N = \{q_{\cdot|y}\}_{y \in [r]} \subset \qm{l}{n}$ we can choose a compatibility-preserving map with $K=1$, $L=g$, $\pi_{x|y,x'} = \mathbf{1}_{x = x'}$ and $\nu_{b|a,x,y,x'} = q_{b|y}$ for all $b \in [l]$, $a \in [k]$, $x,x'\in [g]$ and $y \in [r]$ so that then the resulting multimeter takes the form
    \begin{equation*}
        \sum_{a=1}^k \sum_{x,x' =1}^g \nu_{b|a,x,y,x'} \pi_{x|y,x'} \Gamma^*_{x'}(M_{a|x}) = q_{b|y} \sum_{x=1}^g \Gamma^*_{x}(\id) =   q_{b|y} \id = N_{b|y}
    \end{equation*}
    for all $b \in [l]$ and $y \in [r]$. Thus, $\mathfrak{T}_{ttap} \subset \mathfrak{T}_{cp}$. 

    To see that $\left(\mathfrak{T}_{cp} \cap \mathfrak{T}_{tp} \right) \setminus \left( \mathfrak{T}_{cs} \cup \mathfrak{T}_{c} \cup \mathfrak{T}_{tap} \right) \neq \emptyset$, we can take a compatibility-preserving map with $L=K=1$ (so that the map will automatically be triviality-preserving by Cor.\ \ref{cor:cp-tp}) and choose $n\neq d$ (so that the map cannot be a classical simulation), $l \neq k$ (so that the map cannot be a compression), and the postprocessing $\nu$ to depend on the outcome $a \in [k]$ (so that by Cor.\ \eqref{cor:cp-tap} it is not a trash-and-prepare map). On the other hand, to see that $\left(\mathfrak{T}_{cp} \cap \mathfrak{T}_{tap} \right) \setminus \mathfrak{T}_{tp}  \neq \emptyset$ we refer to Example \ref{ex:cp-tap-ntp} below.

    To finish the proof we still need to show that $\mathfrak{T}_{tap} \setminus \left( \mathfrak{T}_{tp} \cup \mathfrak{T}_{cp}\right)  \neq \emptyset$, $\mathfrak{T}_{tp} \setminus \left( \mathfrak{T}_{cp} \cup \mathfrak{T}_{tap}\right)  \neq \emptyset$ and $\mathfrak{T}_{cp} \setminus \left( \mathfrak{T}_{tp} \cup \mathfrak{T}_{tap} \cup \mathfrak{T}_c \right)  \neq \emptyset$. The first claim follows from the fact that there are trash-and-prepare maps that prepare incompatibile (and thus non-trivial) multimeters so that these maps cannot be triviality-preserving nor compatibility-preserving. For the second claim refer to Example \ref{ex:tp-ncp-ntap} below. The last claim follows from Example \ref{ex:cp-ntp-nc-ntap} below.
\end{proof}

\begin{example}[Compatibility-preserving trash-and-prepare but not triviality-preserving]\label{ex:cp-tap-ntp}
    Let us consider a trash-and-prepare map $\Psi: \M(\complex)_{kdg} \to \M(\complex)_{lnr}$ for which any input multimeter $M=\{M_{\cdot|x}\}_{x \in [g]} \subset \qm{k}{d}$ is mapped to a fixed multimeter $\{N_{\cdot|y}\}_{y \in [r]} \subset \qm{l}{n}$, where $N_{\cdot|y} = N_{\cdot|y'}$ for all $y,y' \in [r]$ for some nontrivial POVM $E := N_{\cdot|y} \in \qm{l}{n}$. It is then clear that the multimeter $N$ is compatible since it only consists of copies of the same POVM. Now $\Psi$ is not triviality-preserving since also trivial multimeters are mapped to $N$ which is nontrivial. On the other hand $\Psi$ is compatibility-preserving: in Eq.\ \eqref{eq:comp-preserving} we can choose $K=1$, $L=l$, $\nu_{b|a,x,y,b'} = \mathbf{1}_{b=b'}$ and define $\Gamma \in \qi{l}{n}{d}$ by setting $\Gamma_b(\varrho) = \Tr[E_b \varrho] \sigma$ for all $\varrho \in \qs{n}$ for some fixed state $\sigma \in \qs{d}$. By plugging these into Eq.\ \eqref{eq:comp-preserving} it is straightforward to see that indeed the transformed multimeter results in the multimeter $N$ which consists only copies of the POVM $E$.
\end{example}

\begin{example}[Triviality-preserving but not compatibility-preserving nor trash-and-prepare]\label{ex:tp-ncp-ntap}
    Let us consider a map $\Psi: \M(\complex)_{kdg} \to \M(\complex)_{lnr}$ in the special case of $s=g=1$ and $d=n=k=l=r=2$ which admits an ancilla-free realization $(\Lambda,\nu)$, where in Eq.\ \eqref{eq:ancilla-free} we choose $\nu_{b|a,y} = \mathbf{1}_{b=a}$ for all $a,b,y \in \{0,1\}$ and we define the channels $\{\Lambda^{(0)}, \Lambda^{(1)}\} \subset \qc{2}{2}$ as $\Lambda^{(y)}(\varrho)= H^y \varrho H^y$ for all $\varrho \in \qs{2}$ for $y \in \{0,1\}$, where $H$ is the Hadamard gate. The transformed POVMs take the following form:
    \begin{equation}
        \sum_{a=0}^1 \nu_{b|a,y} \left(\Lambda^{(y)}\right)^*(M_a) = H^y M_b H^y \subset \qm{2}{2}
    \end{equation}
    for all $b,y \in \{0,1\}$. Now it is clear that the map is triviality-preserving since $HH= I_2$. On the other hand, if we set $M_a = \ketbra{a}{a}$ for all $a \in \{0,1\}$ for the computational basis $\{\ket{0}, \ket{1}\} \subseteq \complex^2$, then although the input multimeter M is compatible (since it consists of only one POVM), the resulting two transformed POVMs given by the previous equation are incompatible. Thus, $\Psi$ is not compatibility-preserving. 
\end{example}

\begin{example}[Compatibility-preserving but not trash-and-prepare nor triviality-preserving nor a compression]\label{ex:cp-ntp-nc-ntap}
    Let us take a compatibility-preserving map and let us set $K=1$ and $g=r=1$ (transforming POVMs to POVMs) so that the transformed POVM $N$ given by Eq.\ \eqref{eq:comp-preserving} takes the form
    \begin{equation}\label{eq:cp-ntp-nc-ntap}
       N^M_b = \sum_{\lambda=1}^L \sum_{a=1}^k \nu_{c|a,\lambda} \Gamma^*_{\lambda}(M_a)
    \end{equation}
    for all $b \in [l]$. We can see that we can make such a compatibility-preserving map not trash-and-prepare by choosing the postprocessing $\nu$ suitably (i.e. such that the condition in Cor.\ \ref{cor:cp-tap} does not hold) nor a compression by choosing $\nu$ and $\Gamma$ such that the above equation will not lead to a channel (we note that compressions between POVMs will always just be  transformations given by channels). Lastly, by choosing $\Gamma$ suitably we can also make the map not triviality-preserving.

    To give an explicit example of such a map, let us take $L=k=l=2$ and $n=d$, and let us take $\Gamma$ to be a Lüders instrument related to some dichotomic POVM $G \in \qm{2}{d}$, which has effects $G_0 = E$ and $G_1 = \id-E$ for some effect $E \in \qe{d}$ such that $E \not\sim \id$, so that 
    $$\Gamma_\lambda(\varrho) = \sqrt{G_\lambda} \varrho  \sqrt{G_\lambda}$$
    for all $\lambda \in \{0,1\}$ and $\varrho \in \qs{d}$. Let us fix the postprocessing $\nu = \{\nu_{\cdot|a,\lambda}\}_{a,\lambda=0}^1$ by setting $\nu_{0|0,1}=\nu_{0|1,1}=0$ and $\nu_{0|0,0} =p$ and $\nu_{0|1,0} =q$ for some $p,q \in [0,1]$ such that $p \neq q$. Now, given an input POVM $M \in \qm{2}{d}$ with effects $M_0 = F$ and $M_1 = \id-F$ for some effect $F \in \qe{d}$, Eq.\ \eqref{eq:cp-ntp-nc-ntap} takes the form for outcome $b=0$ (note that $N^M_1 = \id-N^M_0$)
    \begin{equation*}
        N^F := N^M_0  = p \sqrt{E} F \sqrt{E} + q \sqrt{E} (\id-F) \sqrt{E} = q E + (p-q) \sqrt{E} F \sqrt{E} .
    \end{equation*}
    If we take $F=\pi \id$ for some $\pi \in [0,1]$, we see that
    \begin{equation*}
        N^{\pi \id} = q E + (p-q) \pi E = (p \pi + q (1-\pi))E.
    \end{equation*}
    Now, clearly the map $M \mapsto N^M$ is not triviality-preserving since $E \not\sim \id$, and it is also not trash-and-prepare since by choosing $\pi=0$ we get $N^{0} = q E$ and by choosing $\pi =1$ we get $N^{\id} = p E$, which are not the same effect since $p\neq q$. Lastly, if the map were a compression, in Eq.\ \eqref{eq:compression} we would have to have also $C=1$ so that $N^F = N^M_0 = \Phi(M_0) = \Phi(F)$ for some channel $\Phi \in \qc{d}{d}$. In particular, by taking $F=\id$ we would get $N^\id = \id$ which is a contradiction. Thus, the map cannot be a compression either.
\end{example}

\section{Discussion and outlook} \label{sec:discussion}

In this work, we have thoroughly examined what it means to simulate a set of measurements, i.e. a multimeter, by some other measurements. To this end, we have first characterized transformations between multimeters in terms of their realization involving a preprocessing instrument and a multimeter defined on an ancillary system that is conditioned on the classical inputs and outputs of the simulating multimeter. We have then argued that not all such transformations should be seen as simulations because otherwise one would be able to produce any resource related to multimeters for free. In particular, we argue that from the perspective of an experimenter with a fixed multimeter as measurement apparatus a trash-and-prepare transformation, i.e. the process of discarding a multimeter and replacing it with a fixed multimeter, is not a valid simulation strategy in general, and that simulations should be triviality-preserving, i.e. that no simulation can produce nontrivial multimeters from trivial ones, resulting in our minimal definition of simulability. We characterize these two properties of transformations in terms of their realization in the case when the realization is ancilla-free or only utilizes a classical ancilla. Finally, we demonstrate our findings in the previously proposed simulation scenarios, namely in classical simulation, compressibility, and compatibility-preserving simulations, and compare these simulations to each other and to trash-and-prepare and triviality-preserving transformations.

Some specific open questions arise from trying to generalize our results to the case with a quantum ancilla.  While our realization result for transformations between multimeters is valid in the most general case, when it comes to actual simulation scenarios, our results on triviality-preserving and trash-and-prepare maps are mostly restricted to cases when the realization is ancilla-free or only involves a classical ancilla. This limitation is partially intentional since all the previously proposed simulation scenarios are of this type. However, this does leave open questions regarding the general cases. Some of the difficulties that arise when working with the quantum ancilla results from the lack of fully understanding the degrees of freedom in the realizations of the transformations between multimeters and how different realizations might be connected. Answering these questions requires further research.

Another avenue for future work comes from the fact that although our minimal definition of simulability makes sure that non-triviality of multimeters can be generated in a simulation using trivial multimeters, we are not using a resource theoretic framework. Doing so would likely lead to different notions of simulation depending on the resource at hand. In particular, one could argue, as the authors in \cite{buscemi2020complete} have argued, that also the compatibility of the multimeters should be preserved under simulations so that incompatibility, a fundamental non-classical feature of quantum theory, cannot be produced freely with simulations. An immediate open question there is to look more closely into the proposed compatibility-preserving maps by considering them in full generality as transformations between multimeters and show that these are exactly the type of transformations that preserve compatibility altogether. Moreover, examining other possible resources would be an interesting future direction.

A related but different perspective on multimeter simulation is to consider the preorder that it induces on the set of multimeters: one can say that a multimeter is greater than another multimeter if the latter can be simulated by using the former one. This would create a way to see which multimeters are more useful with respect to each specific simulation task. Subsequently one can start examining questions from the order perspective such as what are the maximal and minimal elements with respect to a particular simulation, what are the simulation irreducible multimeters \cite{filippov2018simulability}, and if there is a natural representative in each equivalence class. We leave this perspective to future work as well.

\bigskip

\noindent\textbf{Acknowledgements.} A.B.\ was supported by the French National Research Agency in the framework of the ``France 2030” program (ANR-11-LABX-0025-01) for the LabEx PERSYVAL. L.L.\ acknowledges support from the European Union’s Horizon 2020 Research and Innovation Programme under the Programme SASPRO 2 COFUND Marie Sklodowska-Curie grant agreement No. 945478 as well as from projects APVV22-0570 (DeQHOST), VEGA 2/0183/21 (DESCOM) and APVV SK-FR-22-0018 (Optimal transport distance for quantum measurements). I.N.~was supported by the ANR projects \href{https://esquisses.math.cnrs.fr/}{ESQuisses}, grant number ANR-20-CE47-0014-01 and \href{https://www.math.univ-toulouse.fr/~gcebron/STARS.php}{STARS}, grant number ANR-20-CE40-0008, and by the PHC programs \emph{Star} (Applications of random matrix theory and abstract harmonic analysis to quantum information theory, grant number 47397VH) and \emph{Stefanik} (Optimal transport distance for quantum measurements, grant number 49882VC).
\bibliographystyle{unsrtnat}
\bibliography{lit}

\newpage

\appendix
\section{Additional properties of quantum channels and instruments}\label{appendix:A}

 One particular advantage of representing the quantum devices introduced in Section \ref{sec:fundamental-quantum-devices} as channels is the fact that then we can treat them all similarly with the tools and representations known for quantum channels (see e.g. \cite{watrous2018theory} for more details). One particularly useful representation of a quantum channel is given by the \emph{Choi–Jamiołkowski isomorphism} \cite{choi1975} which gives us a correspondence between channels and (subset of) states of a higher-dimensional quantum system. More specifically, a linear map $\Phi: \M(\complex)_d \to \M(\complex)_n$ corresponds to a matrix $J_\Phi \in \M(\complex)_{nd}$, called the \emph{Choi matrix of $\Phi$}, defined as 
\begin{equation}\label{eq:Choi-matrix}
    J_\Phi := \sum_{i,j=1}^d \Phi(\ketbra{i}{j}) \otimes \ketbra{i}{j}
\end{equation}
for some orthonormal basis $\{\ket{i}\}_{i=1}^d$ of $\complex^d$. It is know that $\Phi$ is CP if and only if $J_\Phi$ is positive semidefinite, it is TNI if and only if $\Tr_{\complex^n}[J_\Phi] \leq \id_d$ and it is TP if and only if $\Tr_{\complex^n}[J_\Phi] = \id_d$. We denote the set of Choi matrices of channels in $\qc{d}{n}$ by $\mathcal{J}(\complex^{nd}):= \{J_\Phi \in \M(\complex)_{nd} \, : \, \Phi \in \qc{d}{n}\}$ so that in particular $ \frac{1}{d} \mathcal{J}(\complex^{nd}) \subset \qs{nd}$. Conversely, a matrix $J \in \M(\complex)_{nd}$ defines a linear map $\mathcal{E}_J: \M(\complex)_d \to \M(\complex)_n$ by setting 
\begin{equation}\label{eq:Choi-map}
    \mathcal{E}_J(X) := \Tr_{\complex^n}[(\id_n \otimes X^T)J]
\end{equation}
for all $X \in \M(\complex)_d$ where the transpose is taken with respect to the same basis $\{\ket{i}\}_{i=1}^d$ of $\complex^d$. The Choi–Jamiołkowski isomorphism states that $\mathcal{E}_{J_\Phi} = \Phi$ for all $\Phi: \M(\complex)_d \to \M(\complex)_n$ and $J_{\mathcal{E}_J} = J$ for all $J \in \M(\complex)_{nd}$.

Another useful representation for CP maps is the \emph{Stinespring dilation} \cite{stinespring1955}: for any CP map $\Phi:  \M(\complex)_d \to \M(\complex)_n$ and for any $s \geq \rank(J_\Phi) $, there exists an ancillary system $\complex^s$  and a bounded operator $V: \complex^d \to \complex^n \otimes \complex^s$ such that 
\begin{equation}
    \Phi(X) = \Tr_{\complex^s}[V X V^*]
\end{equation}
for all $X \in \M(\complex)_d$. Equivalently, for the dual map $\Phi^*: \M(\complex)_n \to \M(\complex)_d$ defined as $\Tr[\Phi^*(N) D] = \Tr[N \Phi(D)]$ for all $D \in  \M(\complex)_d$ and $N \in  \M(\complex)_n$, this means that
\begin{equation}
    \Phi^*(A) = V^*(A \otimes \id_s) V
\end{equation}
for all $A \in \M(\complex)_n$. In this case, we refer to the tuple $(\complex^s,V)$ as a Stinespring dilation of $\Phi$ (or $\Phi^*$). Furthermore, we say that a dilation $(\complex^s,V)$ is \emph{minimal} if $s = \rank(J_\Phi)$, or equivalently $\complex^n \otimes \complex^s = \{(A \otimes \id_s)V \phi \, : \, A \in \M(\complex)_n, \ \varphi \in \complex^d\}$. It is known that a minimal dilation always exists and that any dilation $(\complex^{s'},V')$ of $\Phi$ is related to a minimal dilation $(\complex^s,V)$ by an isometry $W: \complex^s \to \complex^{s'}$ such that $V' = (\id_n \otimes W)V$, and that $(\complex^{s'},V')$ is minimal if and only if $W$ is unitary \cite[Sec.\ 2.4]{chiribella2009realization}. Now for any dilation of $(\complex^s, V)$ of $\Phi$ it follows that $\Phi$ is TNI if and only if $V^*V \leq \id_d$ and it is TP, or equivalently $\Phi^*$ is \emph{unital} meaning that $\Phi^*(\id_n) = \id_d$, if and only if $V$ is an isometry, i.e., $V^*V=\id_d$. 

Related to the Stinespring dilation of CP maps we will be using the following Radon-Nikodym theorem for CP maps \cite{raginsky2003radon}:

\begin{thm}[{\cite[Thm.~III.3]{raginsky2003radon}}]\label{thm:radon-nikodym-instruments}
Let $N \in \mathbb N$ and let $\{\Phi^*_i\}_{i \in [N]}$ be a collection of CP maps from  $\mathcal M(\complex)_d$ to $\mathcal M(\complex)_n$ such that $\Phi^* := \sum{\Phi_i}^*$ is a CP map with minimal Stinespring dilation $\Phi^*(A) = V^\ast (A \otimes \id_s) V$ for all $A \in \mathcal M(\complex)_d$. Then, there is $Q \in \qm{N}{s}$ such that 
\begin{equation*}
    \Phi_i^\ast(A) = V^\ast (A \otimes Q_i) V \qquad \forall A \in \mathcal M(\complex)_d.
\end{equation*}
\end{thm}

We will now show that we can drop the requirement that the Stinespring dilation is minimal in the above theorem (see also \cite{Wolf2012}).
\begin{cor}\label{cor:radon-nikodym-instruments}
Let $N \in \mathbb N$ and let $\{\Phi^*_i\}_{i \in [N]}$ be a collection of CP maps from  $\mathcal M(\complex)_d$ to $\mathcal M(\complex)_n$ such that $\Phi^* := \sum{\Phi_i}^*$ is a CP map with Stinespring dilation $\Phi^*(A) = W^\ast (A \otimes \id_{s^\prime}) W$ for all $A \in \mathcal M(\complex)_d$. Then, there is $Q^\prime \in \qm{N}{s^\prime}$ such that 
\begin{equation*}
    \Phi_i^\ast(A) = W^\ast (A \otimes Q^\prime_i) W \qquad \forall A \in \mathcal M(\complex)_d.
\end{equation*}
\end{cor}
\begin{proof}
    Let $\Phi^*(A) = V^\ast (A \otimes \id_s) V$ for all $A \in \mathcal M(\complex)_d$ be a minimal Stinespring dilation of $\Phi^*$. Then, by Thm.\  \ref{thm:radon-nikodym-instruments}, there is $Q \in \qm{N}{s}$ such that 
    \begin{equation*}
    \Phi_i^\ast(A) = V^\ast (A \otimes Q_i) V \qquad \forall A \in \mathcal M(\complex)_d.
\end{equation*}
As explained in \cite[Sec.\ 2.4]{chiribella2009realization}, $s^\prime \geq s$ and there exists an isometry $U: \mathbb C^s \to \mathbb C^{s^\prime}$ such that $W = (\mathds{1} \otimes U) V$. Let $P =U U^\ast$, which is an orthogonal projection because $U$ is an isometry. Then, let $Q_i^\prime := U Q_i U^{\ast} + (\mathds{1}_{s^\prime}-P)/N$ for all $i \in [N]$. As $Q_i^\prime \geq 0$ and 
\begin{equation*}
    \sum_{i = 1}^N Q_i^\prime = (\mathds{1}_{s^\prime}-P) + \sum_{i = 1}^N U Q_i U^{\ast} = \mathds{1}_{s^\prime},
\end{equation*}
indeed $Q^\prime \in \qm{N}{s^\prime}$. Finally, we verify that
\begin{align*}
    W^\ast (A \otimes Q^\prime_i) W & = V^\ast (\mathds{1} \otimes U^\ast) A \otimes (U Q_i U^{\ast} + (\mathds{1}_{s^\prime}-P)/N) (\mathds{1} \otimes U) V \\
    & = V^\ast (A \otimes Q_i) V \\
    & =  \Phi_i^\ast(A)
\end{align*}
for all $i \in [N]$, as $U^\ast P U = U^\ast U U^\ast U = \mathds{1}_s$.
\end{proof}

\section{Proof of Theorem \ref{thm:multimeter-transformation-realization}}\label{appendix:B}

In order to characterize the completely positive maps between multimeters, we need to first take a closer look on Choi matrices of multimeters. Let $\Phi_M \in \M\M(g,k,d) \subset \qc{dg}{k}$ be a multimeter for some number $g$  of $k$-outcome POVMs $M= \{M_{\cdot | x}\}_{x \in [g]} \subset \qm{k}{d}$ on $\complex^d$. Now if take some bases $\{\ket{i}\}_{i \in [g]}$ of $\complex^g$ and $\{\ket{\alpha}\}_{\alpha \in [d]}$ of $\complex^d$, then the Choi matrix $J_{\Phi_M}$ can be written as
\begin{align*}
    J_{\Phi_M} &= \sum_{\alpha, \beta =1}^d \sum_{i,j =1}^g  \Phi_M\left(\ketbra{\alpha}{\beta} \otimes \ketbra{i}{j} \right) \otimes \ketbra{\alpha}{\beta} \otimes \ketbra{i}{j}  \\
    &= \sum_{\alpha, \beta =1}^d \sum_{i,j =1}^g   \left[\sum_{x=1}^g \sum_{a=1}^k \Tr[( M_{a | x} \otimes \ketbra{x}{x} )(\ketbra{\alpha}{\beta} \otimes \ketbra{i}{j})]\ketbra{a}{a} \right] \otimes \ketbra{\alpha}{\beta} \otimes \ketbra{i}{j} \\
    &= \sum_{\alpha, \beta =1}^d \sum_{x=1}^g \sum_{a=1}^k \matrixel{\beta}{M_{a|x}}{\alpha} \ketbra{a}{a} \otimes \ketbra{\alpha}{\beta} \otimes \ketbra{x}{x} \\
    &=  \sum_{x=1}^g \sum_{a=1}^k  \ketbra{a}{a} \otimes \left(  \sum_{\alpha, \beta =1}^d \matrixel{\alpha}{M^T_{a|x}}{\beta} \ketbra{\alpha}{\beta}\right) \otimes \ketbra{x}{x} \\
    &= \sum_{x=1}^g \sum_{a=1}^k  \ketbra{a}{a} \otimes M^T_{a|x} \otimes \ketbra{x}{x},
\end{align*}
where the transpose $ M^T_{a|x}$ of $M_{a|x}$ is taken with respect to the basis $\{\ket{\alpha}\}_{\alpha \in [d]}$. Thus, now we have that 
\begin{align}
    \mathcal{J}(\M\M(g,k,d)) &= \left\lbrace \sum_{x=1}^g \sum_{a=1}^k  \ketbra{a}{a} \otimes M^T_{a|x} \otimes \ketbra{x}{x} \, : \, \{M_{\cdot|x}\}_{x \in [g]} \subset \qm{k}{d} \right\rbrace \\
    &= \left\lbrace \sum_{x=1}^g \sum_{a=1}^k  \ketbra{a}{a} \otimes M_{a|x} \otimes \ketbra{x}{x} \, : \, \{M_{\cdot|x}\}_{x \in [g]} \subset \qm{k}{d} \right\rbrace.
\end{align}

Now we can prove Theorem \ref{thm:multimeter-transformation-realization}:
\begin{thm*}
Let $\Psi: \M(\complex)_{kdg} \to \M(\complex)_{lnr}$ be a CP map such that $\Psi(\mathcal{J}(\M\M(g,k,d))) \subseteq \mathcal{J}(\M\M(r,l,n))$. Then $\Psi$ has a realization $(\complex^s, \Lambda, B)$, i.e., there exist an ancillary system $\complex^s$ for some $s \in \nat$, CP maps $\Lambda^*_{x|y}: \M(\complex)_{ds} \to \M(\complex)_n$ such that $\Lambda^*_y := \sum_{x \in [g]} \Lambda^*_{x|y}$ is a unital CP (UCP) map for all $y \in [r]$, and a set of POVMs $B = \{B_{\cdot|a,x,y} \}_{a \in [k], x \in [g], y \in [r]} \subset \qm{l}{s}$ such that 
\begin{equation}
    \Psi \left(J_{\Phi_{\left\lbrace M_{\cdot|x}\right\rbrace_{x\in [g]}}} \right) = J_{\Phi_{\left\lbrace \sum_{x=1}^g \sum_{a=1}^k \Lambda^*_{x|y}(M_{a|x} \otimes B_{\cdot |a,x,y}) \right\rbrace_{y \in [r]}}},   
\end{equation}
where $\left\lbrace \sum_{x=1}^g \sum_{a=1}^k \Lambda^*_{x|y}(M_{a|x} \otimes B_{\cdot |a,x,y}) \right\rbrace_{y \in [r]} \subset \qm{l}{n}$ is a set of POVMs for all\\ $\{M_{\cdot|x}\}_{x\in [g]} \subset \qm{k}{d}$.
\end{thm*}
\begin{proof}
Let us start by defining the following CP maps
\begin{table}[h!]
    \centering
    \begin{tabular}{lll}
    $\forall x \in [g]:$ & $P_x: \M(\complex)_{kdg} \to \M(\complex)_{kd},$ &\\\medskip
    &$P_x(W) := \Tr_{\complex^g}[(\id_k \otimes \id_d \otimes \ketbra{x}{x}) W]$ & $\qquad\forall W \in \M(\complex)_{kdg}$, \\
    $\forall a \in [k]:$ & $Q_a: \M(\complex)_{kd} \to \M(\complex)_{d},$&\\\medskip
    & $Q_a(X) := \Tr_{\complex^k}[(\ketbra{a}{a} \otimes \id_d) X]$ & $\qquad\forall X \in \M(\complex)_{kd}$, \\
    $\forall y \in [r]:$ & $R_y: \M(\complex)_{lnr} \to \M(\complex)_{ln},$&\\\medskip
    & $R_y(Y) := \Tr_{\complex^r}[(\id_l \otimes \id_n \otimes \ketbra{y}{y}) Y]$ & $\qquad\forall Y \in \M(\complex)_{lnr}$, \\
    $\forall b \in [l]:$ & $S_b: \M(\complex)_{ln} \to \M(\complex)_{n},$&\\
    & $S_b(Z) := \Tr_{\complex^l}[(\ketbra{b}{b} \otimes \id_n) Z]$ & $\qquad\forall  Z \in \M(\complex)_{ln}$, 
\end{tabular}
\end{table}

\noindent where now $\{\ket{x}\}_{x \in [g]}$ and $\{\ket{a}\}_{a \in [k]}$ are the same bases of $\complex^g$ and $\complex^k$ respectively that are used to define $\mathcal{J}(\M\M(g,k,d))$ and analogously $\{\ket{y}\}_{y \in [r]}$ and $\{\ket{b}\}_{b \in [l]}$ are the same bases of $\complex^r$ and $\complex^l$ respectively that are used to define $\mathcal{J}(\M\M(r,l,n))$. It is straightforward to see that
\begin{align*}
    P_x^*(A) &= A \otimes \ketbra{x}{x} \quad \forall A \in \M(\complex)_{kd}, \quad Q_a^*(B) = \ketbra{a}{a} \otimes B \quad \forall B \in \M(\complex)_{d}, \\
    R_y^*(C) &= C \otimes \ketbra{y}{y} \quad \forall C \in \M(\complex)_{ln}, \quad S_b^*(D) = \ketbra{b}{b} \otimes D \quad \forall D \in \M(\complex)_{n},
\end{align*}
so that for all $x \in [g]$, $y \in [r]$, $a \in [k]$ and $b \in [l]$ we have
\begin{align}
    P_x(J_{\Phi_M}) &= J_{\Phi_{M_{\cdot|x}}}, \quad J_{\Phi_M} = \sum_{x=1}^g P_x^*(J_{\Phi_{M_{\cdot|x}}}),  \quad Q_a(J_{\Phi_{M_{\cdot|x}}}) = M^T_{a|x}, \quad J_{\Phi_{M_{\cdot|x}}} = \sum_{a=1}^k Q_a^*(M^T_{a|x}), \nonumber \\
    R_y(J_{\Phi_N}) &= J_{\Phi_{N_{\cdot|y}}}, \quad J_{\Phi_N} = \sum_{y=1}^r R_y^*(J_{\Phi_{N_{\cdot|y}}}),  \quad S_b(J_{\Phi_{N_{\cdot|y}}}) = N^T_{b|y}, \quad J_{\Phi_{N_{\cdot|y}}} = \sum_{b=1}^l S_b^*(N^T_{b|y}), \label{eq:PQRS-maps}
\end{align}
for all $M= \{M_{\cdot | x}\}_{x \in [g]} \subset \qm{k}{d}$ and $N= \{N_{\cdot | y}\}_{y \in [r]} \subset \qm{l}{n}$.

For all $x \in [g]$, $y \in [r]$, $a \in [k]$ and $b \in [l]$ we define CP maps $\Psi^*_{b,x|a,y}: \M(\complex)_d \to \M(\complex)_n$ by setting $\Psi^*_{b,x|a,y}:= S_b \circ R_y \circ \Psi \circ P_x^* \circ Q_a^*$. By using the properties of the above defined maps and the fact that $\Psi(\mathcal{J}(\M\M(g,k,d))) \subseteq \mathcal{J}(\M\M(r,l,n))$ it is straightforward to check that
\begin{align}
     \sum_{x=1}^g \sum_{a=1}^k \sum_{y=1}^r \sum_{b=1}^l\left(R_y^* \circ S_b^* \circ \Psi^*_{b,x|a,y} \circ Q_a \circ P_x\right)(J_{\Phi_{M}}) &= \Psi(J_{\Phi_{M}}) 
\end{align}
for all $M=\{M_{\cdot|x}\}_{x\in [g]} \subset \qm{k}{d}$. 

Let us now focus on the properties of the maps $\Psi^*_{b,x|a,y}$. First, since $\Psi(\mathcal{J}(\M\M(g,k,d))) \subseteq \mathcal{J}(\M\M(r,l,n))$, from the properties of the maps $P_x$, $Q_a$, $R_y$ and $S_b$ it follows that
\begin{equation}
    \left\lbrace   \sum_{x=1}^g \sum_{a=1}^k \Psi^*_{\cdot,x|a,y}(M^T_{a|x}) \right\rbrace_{y \in [r]} \subset \qm{l}{n}
\end{equation}
for all $M = \{M_{\cdot|x}\}_{x \in [g]} \subset \qm{k}{d}$. In particular, now we have that 
\begin{equation}
    \sum_{b=1}^l\sum_{x=1}^g \sum_{a=1}^k \Psi^*_{b,x|a,y}(M^T_{a|x}) = \id_n
\end{equation}
for all $y \in [r]$ for all $M = \{M_{\cdot|x}\}_{x \in [g]} \subset \qm{k}{d}$.

Let us now take some $(a_1, \ldots, a_g) \in [k]^g$ and define POVMs $A = \{A_{\cdot|x}\}_{x \in [g]} \subset \qm{k}{d}$ by setting $A_{a|x}= \id_d$, only if $a=a_x$, and naturally due to normalization $A_{a|x}= 0$ otherwise. Now from the above equation we see that
\begin{equation}
    \sum_{b=1}^l\sum_{x=1}^g \sum_{a=1}^k \Psi^*_{b,x|a,y}(A^T_{a|x}) = \sum_{b=1}^l\sum_{x=1}^g \Psi^*_{b,x|a_x,y}(\id_d) = \id_n.
\end{equation}
Hence, since $(a_1, \ldots, a_g)$ was chosen arbitrarily, we have that $\sum_{b=1}^l\sum_{x=1}^g \Psi^*_{b,x|a,y}$ is unital for all $a \in [k]$ and $y \in [r]$.

On the other hand, let us now fix some effect operators $B_x \in \qe{d}$ for all $x \in [g]$. Now if we take some $(a_1, \ldots, a_g), (a'_1, \ldots, a'_g) \in [k]^g$ such that $a_x \neq a'_x$ for all $x \in [g]$ and define POVMs $B = \{B_{\cdot|x}\}_{x \in [g]} \subset \qm{k}{d}$ by setting $B_{a_x|x}= B_x$, $B_{a'_x|x} = \id_d-B_x$ and $B_{a|x}= 0$ otherwise, we see that
\begin{align*}
    \id_n &= \sum_{b=1}^l\sum_{x=1}^g \sum_{a=1}^k \Psi^*_{b,x|a,y}(B^T_{a|x}) = \sum_{b=1}^l\sum_{x=1}^g \Psi^*_{b,x|a_x,y}(B^T_x)+\sum_{b=1}^l\sum_{x=1}^g \Psi^*_{b,x|a'_x,y}(\id_d -B^T_x) \\
    &= \id_n + \sum_{b=1}^l\sum_{x=1}^g \left(\Psi^*_{b,x|a_x,y} - \Psi^*_{b,x|a'_x,y}\right)(B^T_x) \\&= \id_n +  \sum_{x=1}^g \left( \sum_{b=1}^l\Psi^*_{b,x|a_x,y} - \sum_{b=1}^l\Psi^*_{b,x|a'_x,y}\right)(B^T_x).
\end{align*}
Now if we fix $x' \in [g]$ and take $B_{x'}=B$ for some $B \in \qe{d}$ and $B_x =0$ otherwise, we then must have that $ \left( \sum_{b=1}^l\Psi^*_{b,x'|a_{x'},y} - \sum_{b=1}^l\Psi^*_{b,x'|a'_{x'},y}\right)(B^T) = 0$  for all $y \in [r]$. Since $B$ can be chosen arbitrarily and since the set of effects $\qe{d}$, as well as their transposes, spans $\M(\complex)_d$, we thus have that $\sum_{b=1}^l\Psi^*_{b,x'|a_{x'},y} =\sum_{b=1}^l\Psi^*_{b,x'|a'_{x'},y}$ for all $y \in [r]$. Finally, since $x' \in [g]$ and the outcomes $(a_1, \ldots, a_g), (a'_1, \ldots, a'_g) \in [k]^g$ can also be chosen arbitrarily, we conclude that
$\sum_{b=1}^l\Psi^*_{b,x|a,y} = \sum_{b=1}^l\Psi^*_{b,x|a',y}$ for all $x \in [g]$, $y \in [r]$ and $a,a' \in [k]$. Thus, we may denote $\Psi^*_{x|y} := \sum_{b=1}^l\Psi^*_{b,x|a,y}$.

To summarize so far, we have the following properties
\begin{align}
    & \Psi^*_{b,x|a,y}: \M(\complex)_d \to \M(\complex)_n  \quad \mathrm{CP} \quad \forall x \in [g], \ \forall y \in [r], \ \forall a \in [k], \ \forall b \in [l] \\ 
    & \Psi^*_{x|y} := \sum_{b=1}^l\Psi^*_{b,x|a,y} = \sum_{b=1}^l\Psi^*_{b,x|a',y}  \quad \mathrm{CP} \quad \forall x \in [g], \ \forall y \in [r], \ \forall a,a' \in [k]\\  
    &  \Psi^*_y := \sum_{x=1}^g \Psi^*_{x|y}  \quad \mathrm{UCP} \quad \forall y \in [r],
\end{align}
where UCP means that the maps are unital and CP.

Let us now take $s \geq \max_{x,y} \rank(J_{\Psi^*_{x|y}})$ and a Stinespring dilation $(\complex^s, V_{xy})$ for $\Psi^*_{x|y}$ so that $\Psi^*_{x|y}(A) = V^*_{xy}(A \otimes \id_s)V_{xy}$ for all $A \in \M(\complex)_d$ and such that $V^*_{xy} V_{xy} \leq \id_d$ for all $x \in [g]$ and $y \in [r]$. Furthermore, since $\Psi^*_y = \sum_{x} \Psi^*_{x|y}$ is unital it follows that $\sum_x V^*_{xy}V_{xy} = \id_d$ for all $y \in [r]$. Since $\Psi^*_{x|y} = \sum_{b} \Psi^*_{b,x|y,a}$ for all $a \in [k]$, $x \in [g]$ and $y \in [r]$, by Cor.~\ref{cor:radon-nikodym-instruments} there exists POVMs $\tilde B=\{\tilde B_{\cdot |a,x,y}\}_{a \in [k], x \in [g], y \in [r]} \subset \qm{l}{s}$ such that 
\begin{equation}
    \Psi^*_{b,x|a,y}(A) = V^*_{xy}(A \otimes B_{b|a,x,y} ) V_{xy}\, .
\end{equation}

We can now define CP maps $\tilde \Lambda^*_{x|y}: \M(\complex)_{ds} \to \M(\complex)_n$ by setting $\tilde \Lambda^*_{x|y}(X) = V^*_{xy} X V_{xy}$ for all $x \in [g]$ and $y \in [r]$. From the fact that  $\sum_x V^*_{xy}V_{xy} = \id_d$ it follows that $\tilde \Lambda^*_y := \sum_x \tilde \Lambda^*_{x|y}$ is UCP for all $y \in [r]$. Thus, we have shown that there exist an ancillary system $\complex^s$, CP maps $\tilde \Lambda^*_{x|y}: \M(\complex)_{ds} \to \M(\complex)_n$ such that $\tilde \Lambda^*_y := \sum_{x \in [g]} \tilde \Lambda^*_{x|y}$ is  UCP for all $y \in [r]$, and a set of POVMs $\tilde B = \{\tilde B_{\cdot|a,x,y} \}_{a \in [k], x \in [g], y \in [r]} \subset \qm{l}{s}$ such that 
\begin{equation}
    \Psi^*_{b,x|a,y}(A) = \tilde \Lambda^*_{x|y}(A \otimes \tilde B_{b|a,x,y})
\end{equation}
for all $x \in [g]$, $y \in [r]$, $a \in [k]$ and $b \in [l]$.

Finally, we now see that for any $M = \{M_{\cdot|x}\}_{x \in [g]} \subset \qm{k}{d}$ we have that
\begin{align}
    \Psi(J_{\Phi_M}) &=  \sum_{x=1}^g \sum_{a=1}^k \sum_{y=1}^r \sum_{b=1}^l\left(R_y^* \circ S_b^* \circ \Psi^*_{b,x|a,y} \circ Q_a \circ P_x\right)(J_{\Phi_M})  \\
    &= \sum_{x=1}^g \sum_{a=1}^k \sum_{y=1}^r \sum_{b=1}^l\left(R_y^* \circ S_b^* \right) \left( \Psi^*_{b,x|a,y}(M^T_{a|x})  \right)  \\
    &=  \sum_{y=1}^r \sum_{b=1}^l\left(R_y^* \circ S_b^* \right)\left( \sum_{x=1}^g \sum_{a=1}^k \tilde \Lambda^*_{x|y}(M^T_{a|x} \otimes \tilde B_{b|a,x,y}) \right) \\
    &=  \sum_{y=1}^r R_y^* \left(J_{\Phi_{(\sum_{x=1}^g \sum_{a=1}^k \tilde \Lambda^*_{x|y}(M^T_{a|x} \otimes \tilde B_{\cdot |a,x,y}))^T}} \right) \\
    &= J_{\Phi_{\left\lbrace \sum_{x=1}^g \sum_{a=1}^k (\tilde \Lambda^*_{x|y}(M^T_{a|x} \otimes \tilde B_{\cdot |a,x,y}))^T \right\rbrace_{y \in [r]}}},
\end{align}
where the outer transposition in the last line is taken in the $\{\ket{\beta}\}_{\beta \in [n]}$ basis of $\complex^n$. Now, we can define a new CP map $\Lambda^*_{x|y}(X) = (\tilde \Lambda^*_{x|y}(X^T))^T$, where the outer transposition is in the $\{\ket{\beta}\}_{\beta \in [n]}$ basis and the inner in the $\{\ket{\alpha} \ket{c}\}_{\alpha \in [d], c \in [s]}$ basis, with $\{\ket{c}\}_{c \in [s]}$ being a basis of $\complex^s$. Moreover, we define a  new UCP map $\Lambda^*_{y} = \sum_{x \in [g]} \Lambda^*_{x|y}$ and a new POVM $ B=\{\tilde B_{\cdot |a,x,y}^T\}_{a \in [k], x \in [g], y \in [r]}$, where the transpose is in the $\{\ket{c}\}_{c \in [s]}$ basis. Thus,
\begin{equation*}
     \Psi \left(J_{\Phi_{\left\lbrace M_{\cdot|x}\right\rbrace_{x\in [g]}}} \right) = J_{\Phi_{\left\lbrace \sum_{x=1}^g \sum_{a=1}^k \Lambda^*_{x|y}(M_{a|x} \otimes B_{\cdot |a,x,y}) \right\rbrace_{y \in [r]}}}
\end{equation*}
and $\left\lbrace \sum_{x=1}^g \sum_{a=1}^k \Lambda^*_{x|y}(M_{a|x} \otimes B_{\cdot |a,x,y}) \right\rbrace_{y \in [r]} = \left\lbrace \sum_{x=1}^g \sum_{a=1}^k (\Psi^*_{\cdot,x|a,y}(M^T_{a|x}))^T \right\rbrace_{y \in [r]}   \subset \qm{l}{n}$.
\end{proof}

\end{document}